\documentclass[journal]{IEEEtran}

\usepackage{cite}
\usepackage{amsmath,amssymb,amsfonts}
\usepackage{amsthm}
\usepackage{algorithmic}
\usepackage{algorithm}
\usepackage{graphicx}
\usepackage{textcomp}
\usepackage{xcolor}
\usepackage{booktabs}
\usepackage{subcaption}
\usepackage{url}
\usepackage{bm}
\usepackage{tikz}
\usepackage{ulem}     
\usepackage{enumitem} 
\usetikzlibrary{arrows.meta,positioning,calc,shapes.geometric,decorations.pathreplacing}

\newif\ifshowchanges
\newif\ifblueonly
\showchangesfalse  
\blueonlyfalse

\ifshowchanges
  \ifblueonly
    \long\def\new#1{{\color{blue}#1}}
    \long\def\del#1{}
    \newcommand{\revtag}[1]{}
    \newcommand{\revnote}[1]{}
  \else
    \long\def\new#1{{\color{blue}#1}}
    \long\def\del#1{{\color{red}#1}}
    \newcommand{\revtag}[1]{\textcolor{red}{\textsuperscript{\textbf{\,[#1]}}}}
    \newcommand{\revnote}[1]{\marginpar{\footnotesize\color{red}#1}}
  \fi
\else
  \long\def\new#1{#1}
  \long\def\del#1{}
  \newcommand{\revtag}[1]{}
  \newcommand{\revnote}[1]{}
\fi

\graphicspath{{matlab/}{./matlab/}{./}{figs/}}

\DeclareMathOperator*{\argmin}{arg\,min}
\DeclareMathOperator{\diag}{diag}

\newcommand{\bA}{\mathbf{A}}

\newcommand{\bH}{\mathbf{H}}
\newcommand{\bM}{\mathbf{M}}
\newcommand{\bU}{\mathbf{U}}
\newcommand{\bx}{\mathbf{x}}
\newcommand{\by}{\mathbf{y}}
\newcommand{\bz}{\mathbf{z}}
\newcommand{\bw}{\mathbf{w}}
\newcommand{\bj}{\mathbf{j}}
\newcommand{\bb}{\mathbf{b}}
\newcommand{\bh}{\mathbf{h}}
\newcommand{\bI}{\mathbf{I}}

\newcommand{\calA}{\mathcal{A}}
\newcommand{\calJ}{\mathcal{J}}
\newcommand{\calK}{\mathcal{K}}
\newcommand{\calN}{\mathcal{N}}
\newcommand{\calS}{\mathcal{S}}

\newcommand{\calCN}{\mathcal{CN}}
\newcommand{\bbE}{\mathbb{E}}
\newcommand{\bbR}{\mathbb{R}}
\newcommand{\bbC}{\mathbb{C}}

\newtheorem{theorem}{Theorem}
\newtheorem{proposition}{Proposition}
\newtheorem{corollary}{Corollary}

\newtheorem{remark}{Remark}
\newtheorem{definition}{Definition}

\begin{document}

\title{Jamming-Resilient OTFS-NOMA\\
\large (Extended Version with Full Proofs and Numerical Appendix)}

\author{Michel~Kulhandjian,~\IEEEmembership{Senior~Member,~IEEE,}
        Hovannes~Kulhandjian,~\IEEEmembership{Senior~Member,~IEEE,}
        and~Theodoros~A.~Tsiftsis,~\IEEEmembership{Senior~Member,~IEEE}%
\thanks{M.~Kulhandjian is with the Department of Electrical and Computer Engineering, Rice University, Houston, TX 77005 USA (e-mail: \texttt{mkulhandjian@outlook.com}).}%
\thanks{H.~Kulhandjian is with the Department of Electrical and Computer Engineering, California State University, Fresno, CA 93740 USA.}%
\thanks{T.~A.~Tsiftsis is with the Department of Informatics and Telecommunications, University of Thessaly, Lamia 35100, Greece (e-mail: \texttt{tsiftsis@uth.gr}).}}

\markboth{IEEE Transactions on Wireless Communications}%
{Kulhandjian~\MakeLowercase{\textit{et~al.}}: Jamming-Resilient OTFS-NOMA}

\maketitle

\begin{abstract}
In this paper we propose a jamming-resilient sparse
delay-Doppler NOMA scheme for OTFS and derive closed-form BER
expressions from a Marchenko--Pastur conditioning argument on
random unitary submatrices, in contrast to the well-known
partial-band error floor of conventional OTFS-NOMA. The resulting
bound is Hadamard-DFT-Haar invariant and remains valid under
compromise of the shared randomization seed. A superincreasing
power allocation brings low-complexity per-bin SIC up to the
maximum-likelihood detector on the composite constellation, and
a pairwise cluster rule scales the design beyond four users.
Compared against a fair-baseline, the proposed scheme is the only
configuration without a jammer-induced error floor and gives a
$\approx 6$~dB horizontal SNR gain at BER~$=10^{-2}$.
\end{abstract}


\section{Introduction}
\label{sec:introduction}

\subsection{Background}
High-mobility wireless applications such as high-speed rail, unmanned
aerial vehicles, vehicle-to-everything (V2X), and low-Earth-orbit (LEO)
satellite links demand reliable connectivity under simultaneously time-varying and
frequency-selective channel conditions~\cite{wang2020,zhang2019}. Conventional
orthogonal frequency-division multiplexing (OFDM) suffers severe intercarrier
interference in such doubly-dispersive environments, motivating the
introduction of OTFS
modulation~\cite{hadani2017otfs,raviteja2018interference}, which places information
symbols on the two-dimensional delay-Doppler (DD) grid and converts the
fast-varying time-frequency channel into a quasi-stationary DD-domain channel.
To meet the growing demand for massive connectivity, OTFS has been combined
with NOMA, giving rise to OTFS-NOMA in both
power-domain~\cite{ding2019otfsnoma,ding2020robust} and code-domain (e.g.,
sparse code multiple access (SCMA)) variants~\cite{deka2021otfsscma,kulhandjian2024otfsnoma}.

Increasingly, however, these systems are deployed in adversarial environments
where intentional jamming is the dominant impairment. Tactical communications,
satellite uplinks under spoofing attacks, V2X under coexistence interference,
and contested-spectrum operation all share the common threat model: an
adversary injects high-power interference into a subset of the time, frequency,
or delay-Doppler resources, with the goal of denying or degrading the legitimate
link~\cite{poisel2011}. While DD-domain modulation provides intrinsic
robustness against Doppler-induced impairments, it does not provide any
inherent protection against pattern-aware jammers~\cite{zhuo2026,deng2023otfsscma}.

This paper closes that gap with a five-ingredient stack tightly
co-designed for the jamming-resilient multi-user setting: (i)~\emph{sparse}
data placement on the DD grid, (ii)~\emph{unitary spreading} across all
$N_b$ bins (Hadamard transform (HT), discrete Fourier transform (DFT),
or random orthogonal matrix), (iii)~\emph{per-frame randomization} of
the active bin subset via a shared pseudo-random seed (the ``C5''
protocol), (iv)~\emph{superincreasing power allocation} (PA) that
makes the $O(K)$ successive interference cancellation (SIC) decoder
maximum-likelihood (ML)-optimal on the NOMA composite constellation,
and (v)~\emph{disjoint pairwise clustering} when more than four users
are present. The receiver
threshold-excises jammer-contaminated bins, recovers the sparse
signal by least squares, and decodes per cluster by SIC. The resulting
architecture has no jammer-induced error floor (in contrast to
conventional OTFS-NOMA's irreducible $\rho_J/2$ floor), defeats
pattern-aware adversaries even under seed compromise, and scales to
$K_{\rm tot}\le n_a/2$ users with floor bit-error rate (BER) at the
recommended signal-to-noise ratio (SNR) operating point.

\subsection{Prior Work on OTFS Anti-Jamming and OTFS-NOMA}

Single-user anti-jamming OTFS has focused on receiver-side suppression
via forward-consecutive mean excision
(FCME)~\cite{li2025otfs_jamming},
energy concentration~\cite{sciencedirect2024_nbi}, and iterative
interference cancellation (IC). Zhuo and Qiu~\cite{zhuo2026} proposed
H-OTFS (sparse DD placement, Hadamard spread, threshold excision,
alternating-direction method of multipliers (ADMM) compressed-sensing
(CS) recovery), but the scheme is single-user and does not exploit
power-domain NOMA. In the multi-user OTFS-NOMA jamming setting, the
closest prior work is Deng--Ge--Ding~\cite{deng2023otfsscma}, which
permutes SCMA user groups across fixed
delay/Doppler partitions under narrowband-interference /
periodic-impulse-noise (NBI/PIN) models, and the resource-block (RB)
SCMA hopping of Yang \emph{et al.}~\cite{lu2023interRBhopping}; both
differ from ours on power-vs-code domain, bin-vs-RB hopping,
fixed-vs-arbitrary placement, adversarial-vs-benign jammer,
turbo-vs-one-shot recovery, and simulation-vs-analytical scope (full
breakdown in Section~\ref{sec:discussion}).

Classical (non-DD) NOMA anti-jamming uses game-theoretic
PA~\cite{warwick2023_satnoma}, reactive-jammer
bypass~\cite{arxiv2024_reactive}, and intelligent reflecting surface
(IRS)-aided schemes~\cite{najimi2025_irsnoma}, none with DD-channel
modeling. Hadamard-NOMA for benign channels studies channel state
information (CSI)
robustness~\cite{usman2018joint,baig2018papr,bouslam2024noma,arxiv2603_spectralht}
and peak-to-average power ratio (PAPR)
reduction~\cite{baig2018papr}; DFT-spread OTFS-NOMA targets
PAPR/positioning~\cite{dftspread_otfs_noma_2024}; none jointly
addresses sparse-DD placement, jamming, and a closed-form design rule.

The OFDM pilot-randomization
lineage~\cite{clancy2011ofdm,lapan2013pilot} established the
principle for pilots; our C5 protocol extends it from pilots to data
placement, from single-user to multi-user power NOMA, and from
heuristics to a Marchenko--Pastur (M-P) conditioning framework.
Adjacent multi-user DD
work~\cite{thomas2025_zakotfs_uplink,otfs_noma_mimo_2024,krishna2026noma_comparison}
does not consider adversarial threats.
\new{\revtag{R3.5}Cross-domain joint channel-and-data estimation
for mmWave MIMO-OFDM with low-resolution
ADCs~\cite{jstsp2025_crossdomain} is a complementary line of work:
it addresses hardware-limited receiver front-ends through
joint-domain estimation rather than adversarial jamming, and its
estimation framework could in principle be composed with our C5
randomization to harden hardware-constrained OTFS-NOMA receivers.}

\subsection{Motivation}
\label{sec:motivation}

\new{\revtag{R1.1}Three limitations in the existing literature
motivate this work. First, existing OTFS anti-jamming schemes
(FCME, energy concentration, iterative IC, ADMM
recovery~\cite{li2025otfs_jamming,sciencedirect2024_nbi,zhuo2026})
assume the jammer has known spectral or temporal structure (NBI,
PIN, deterministic patterns) and tune the receiver accordingly;
they do not provide guarantees against adaptive or pattern-aware
adversaries. Second, prior OTFS-NOMA schemes, both
power-domain~\cite{ding2019otfsnoma,ding2020robust} and
code-domain~\cite{deka2021otfsscma,kulhandjian2024otfsnoma,deng2023otfsscma,lu2023interRBhopping},
allocate users to fixed delay/Doppler resources (bins,
partitions, or RBs). A jammer that learns the allocation (for
example by observing channel-estimation pilots or through
feedback leakage) can target it exactly, and no existing
closed-form framework quantifies the resulting BER floor or the
gain from randomization. Third, random-matrix concentration is
well studied for pilot-randomized OFDM and compressed sensing
but has not been connected to NOMA-SIC decoding or to
superincreasing-PA constellations; a closed-form bound on Bob's
BER that is simultaneously a function of the jamming fraction
$\rho_J$ and the sparsity $n_a/N_b$, valid under seed
compromise, and tight at operating sparsity is missing. The
present work addresses these three limitations with a unified
architecture: pattern-aware randomization, a closed-form M-P BER
bound, and a defense-in-depth property under seed compromise.}

\subsection{Contributions}

\del{This paper proposes a unified architecture for jamming-resilient
multi-user DD-domain communication. The seven contributions below
form a single integrated stack delivering $\ge 40$~dB BER-ratio
improvement against pattern-aware adversaries and floor BER at
$K_{\rm tot}=8$ users by SNR$=20$~dB:}%
\new{\revtag{R1.2, R3.7.2}%
This paper proposes a unified architecture for jamming-resilient
multi-user delay-Doppler communication. The seven contributions
summarized below form a single integrated stack; detailed
mathematical formulations are deferred to the sections cited
inside each bullet.}

\begin{itemize}
\item[C1.] User data is placed on a small random
subset of delay-Doppler bins and then spread over all bins by a
unitary precoder. The spreading step hides which bins actually
carry information, so the jammer cannot preferentially target
them. Hadamard, discrete-Fourier, and Haar-random precoders are
all admissible choices~(Sec.~II-D).

\item[C2.] A three-stage receiver that
(a)~detects and discards jammer-contaminated bins by a
power-anomaly threshold, (b)~recovers the sparse signal on the
kept bins by least squares, and (c)~decodes user bits by
successive interference cancellation. Complexity is linear in the
grid size for the first stage and cubic in the sparse-signal
dimension for the second~(Sec.~II-F).

\item[C3.] A Marchenko--Pastur concentration argument on
random unitary submatrices gives a closed-form BER expression
that has no jammer-induced error floor, in contrast to the
irreducible floor of conventional partial-band
jamming~(Theorem~\ref{thm:prop_ber}).

\item[C4.] The same conditioning argument gives a
closed-form rule that, for any target penalty budget, specifies
the maximum admissible loading in terms of the jamming
fraction~(Table~\ref{tab:design_rule}).

\item[C5.] The active set is redrawn every frame from a
pseudo-random seed shared between the transmitter and the
receiver, so the jammer never sees the same pattern twice. We
prove that even if the seed is compromised, the random-subset
structure of the active set still gives floor BER by
Marchenko--Pastur universality, so secrecy can fail without the
system breaking~(Sec.~\ref{sec:c5_protocol}).

\item[C6.] User powers are allocated by the
Merkle--Hellman superincreasing recurrence (the algebraic
property used in knapsack cryptography). Two consequences:
low-complexity SIC becomes bit-exactly equivalent to ML on the
NOMA composite constellation, and the usual SIC
error-propagation ceiling that limits more-than-two-user decoding
is removed~(Sec.~\ref{sec:super_pa},~\cite{kulhandjian2026_hnoma_wcl}).

\item[C7.] For more than four users, the system
partitions users into pairs and assigns each pair its own
disjoint bin subset. This OMA-friendly NOMA rule is optimal
under the proposed receiver and the stated BER objective, and
reaches floor BER for up to eight users at moderate
SNR~(Sec.~\ref{sec:cluster_design}).
\end{itemize}

The remainder is organized as: system model (Sec.~\ref{sec:system_model}),
PA framework (Sec.~\ref{sec:super_pa}), analytical BER
(Sec.~\ref{sec:analytical}), C5 protocol (Sec.~\ref{sec:c5_protocol}),
cluster design (Sec.~\ref{sec:cluster_design}), numerical results
(Sec.~\ref{sec:numerical}), and discussion
(Sec.~\ref{sec:discussion}).

\textit{Notation:} bold lowercase/uppercase for vectors/matrices;
$\calCN(\mu,\sigma^2)$ for complex Gaussian; $Q(\cdot)$ for the
Gaussian Q-function; $\bA_{\calK,\calA}$ for the row-$\calK$,
column-$\calA$ submatrix; $\sigma_{\min}(\bA)$ for smallest singular
value.
\new{\revtag{R3.8.2}We reserve $M$ and $N$ for the delay and
Doppler dimensions of the OTFS grid, and use $N_b\triangleq MN$
exclusively for the total DD-bin count; any claim about the
``grid size'' in this paper is in $N_b$.}

\section{System Model}
\label{sec:system_model}

Fig.~\ref{fig:system_diagram} summarizes the end-to-end transmit-receive
chain of the proposed sparse-DD NOMA architecture, from per-user bits to
detected bits, with the key architectural ingredients (sparse placement,
unitary precoding, C5 randomization, and excision-LS-SIC) explicitly
shown.

\begin{figure*}[t]
\centering
\resizebox{0.96\textwidth}{!}{%
\begin{tikzpicture}[
  >={Stealth[length=2.5mm]},
  block/.style={rectangle,draw,minimum width=14mm,minimum height=8mm,
                align=center,font=\small},
  smallblock/.style={rectangle,draw,minimum width=10mm,minimum height=6mm,
                align=center,font=\footnotesize},
  oper/.style={circle,draw,minimum size=6mm,inner sep=1pt,font=\small},
  dot/.style={circle,fill,inner sep=0.6pt},
  node distance=4mm and 6mm,
  font=\footnotesize
]
\node[block] (b1)  {$b_1$};
\node[block,below=2mm of b1] (b2)  {$b_2$};
\node[font=\small,below=2mm of b2] (bdots) {$\vdots$};
\node[block,below=2mm of bdots] (bK)  {$b_K$};

\node[smallblock,right=of b1] (m1)  {BPSK};
\node[smallblock,right=of b2] (m2)  {BPSK};
\node[smallblock,right=of bK] (mK)  {BPSK};

\node[smallblock,right=of m1] (a1)  {$\sqrt{\alpha_1 P_s}$};
\node[smallblock,right=of m2] (a2)  {$\sqrt{\alpha_2 P_s}$};
\node[smallblock,right=of mK] (aK)  {$\sqrt{\alpha_K P_s}$};

\node[oper,right=8mm of a2]   (sum) {$+$};

\node[smallblock,right=10mm of sum,align=center]  (sparse)
  {Sparse\\placement\\on $\calA_t$};

\node[smallblock,right=of sparse,align=center]  (Umtx)
  {Unitary\\$\bU$\\(HT/DFT)};

\node[oper,right=of Umtx]  (txout) {$\by$};

\node[smallblock,right=of txout,align=center,fill=gray!10] (chan)
  {$\diag(\bh)$\\(Rayleigh)};

\node[oper,right=of chan] (plus) {$+$};

\node[smallblock,above=4mm of plus,fill=red!10,align=center] (jam)
  {jammer\\$\bj,\ \calJ_t$};
\node[smallblock,below=4mm of plus,fill=blue!10] (noise) {$\bw$};

\node[smallblock,right=of plus,align=center] (excise)
  {Threshold\\excise};

\node[smallblock,right=of excise,align=center] (ls)
  {LS\\$\hat{\bx}_\calA$};

\node[smallblock,right=of ls,align=center] (sic)
  {SIC\\decode};

\node[block,right=of sic]   (bhat1) {$\hat{b}_1$};
\node[block,below=2mm of bhat1] (bhat2) {$\hat{b}_2$};
\node[font=\small,below=2mm of bhat2] (bhatd) {$\vdots$};
\node[block,below=2mm of bhatd] (bhatK) {$\hat{b}_K$};

\node[smallblock,below=10mm of sparse,fill=green!10,align=center] (seed)
  {Shared seed $s_0$\\(C5 protocol)};

\draw[->] (b1) -- (m1);
\draw[->] (b2) -- (m2);
\draw[->] (bK) -- (mK);
\draw[->] (m1) -- (a1);
\draw[->] (m2) -- (a2);
\draw[->] (mK) -- (aK);
\draw[->] (a1.east) -- (sum);
\draw[->] (a2.east) -- (sum);
\draw[->] (aK.east) -- (sum);
\draw[->] (sum)    -- node[above,font=\scriptsize]{$\bx_\calA$}  (sparse);
\draw[->] (sparse) -- node[above,font=\scriptsize]{$\bx$}        (Umtx);
\draw[->] (Umtx)   -- (txout);
\draw[->] (txout)  -- (chan);
\draw[->] (chan)   -- (plus);
\draw[->] (jam)    -- (plus);
\draw[->] (noise)  -- (plus);

\draw[->,green!50!black,dashed] (seed) -- (sparse.south);
\draw[->,green!50!black,dashed] (seed.east) -| (excise.south);

\draw[->] (plus)   -- node[above,font=\scriptsize]{$\bz$} (excise);
\draw[->] (excise) -- node[above,font=\scriptsize]{$\bz_\calK$} (ls);
\draw[->] (ls)     -- node[above,font=\scriptsize]{$\hat{\bx}_\calA$} (sic);
\draw[->] (sic.east) -- (bhat1.west);
\draw[->] (sic.east) -- (bhat2.west);
\draw[->] (sic.east) -- (bhatK.west);

\draw[decorate,decoration={brace,amplitude=4pt}]
  ($(b1.north west)+(-2mm,2mm)$) --
  ($(aK.north east |- b1.north west)+(2mm,2mm)$)
  node[midway,above=4pt,font=\small\itshape]{NOMA superposition};
\draw[decorate,decoration={brace,amplitude=4pt}]
  ($(sparse.north west)+(-2mm,2mm)$) --
  ($(txout.north east)+(2mm,2mm)$)
  node[midway,above=4pt,font=\small\itshape]{Sparse-DD + unitary spread};
\draw[decorate,decoration={brace,amplitude=4pt}]
  ($(excise.north west)+(-2mm,2mm)$) --
  ($(sic.north east)+(2mm,2mm)$)
  node[midway,above=4pt,font=\small\itshape]{Excision-LS-SIC receiver};
\end{tikzpicture}%
}
\caption{\new{End-to-end system architecture of the proposed
sparse-DD NOMA scheme.}}%
\del{End-to-end system architecture of the proposed sparse-DD
NOMA scheme. $K$ user bits $\{b_k\}$ are BPSK-modulated and scaled by
the superincreasing PA $\sqrt{\alpha_k P_s}$, summed to form the
NOMA-superposed vector $\bx_\calA$, and placed on the per-frame
random active set $\calA_t$ (drawn via the C5 shared-seed protocol,
green dashed arrows). The sparse vector $\bx$ is spread by the unitary
$\bU$ (Hadamard recommended) to $\by$, then transmitted through the
Rayleigh DD-channel with additive jammer $\bj$ on $\calJ_t$ and noise
$\bw$. The receiver, knowing $\calA_t$ via the shared seed, performs
power-threshold excision of jammer-contaminated bins, least-squares
recovery $\hat{\bx}_\calA$ from the kept observations $\bz_\calK$,
and per-bin SIC to obtain $\{\hat{b}_k\}$. The C5 protocol ensures
that the random active set $\calA_t$ is unknown to the jammer in
expectation, while seed compromise (oracle threat J3) is absorbed
by M-P conditioning (Section~\ref{sec:analytical}).}
\label{fig:system_diagram}
\end{figure*}

\subsection{OTFS Modulation in the Delay-Doppler Domain}

A single base station serves $K$ users via OTFS on an $M\times N$
DD grid; $N_b\triangleq MN$ is the total bin count
and $n\in\{1,\ldots,N_b\}$ enumerates bins in raster order. Under
standard integer delay-Doppler indices and bi-orthogonal pulse
shaping~\cite{raviteja2018interference,hadani2017otfs}, the DD
input--output relation is per-bin scalar:
\begin{equation}
\label{eq:dd_bin_model}
z[n] = h[n]\,y[n] + w[n] + j[n], \quad n=1,\ldots,N_b,
\end{equation}
\new{\revtag{R1.5, R3.1}%
where $y[n]\in\bbC$ is the transmitted DD-domain symbol on bin $n$
(the precoded output of~\eqref{eq:unitary_spread} below),
$z[n]\in\bbC$ is the received DD-domain observation on bin $n$,
$h[n]\in\bbC$ is the effective per-bin DD-domain channel coefficient
seen at bin $n$ after the OTFS symplectic Fourier transform, and
$w[n]\sim\calCN(0,N_0)$ is AWGN with per-bin variance $N_0$;
$j[n]\in\bbC$ denotes the jammer contribution on bin $n$,
characterized below.}%
\del{with $w[n]\sim\calCN(0,N_0)$ additive white Gaussian noise (AWGN) and
$j[n]$ the jammer contribution.}

\new{\revtag{R3.1}%
\begin{remark}
\label{rem:dd_assumptions}
The scalar model~\eqref{eq:dd_bin_model} is the standard reduction
of the full doubly-selective OTFS input--output relation under four
assumptions: (i)~ideal bi-orthogonal transmit/receive pulses
(sufficient-statistic sampling at the DD
lattice)~\cite{raviteja2018interference,hadani2017otfs};
(ii)~a cyclic prefix sufficient to absorb the maximum delay spread,
so that inter-frame interference is zero; (iii)~integer delay and
integer Doppler taps, so that the DD channel does not couple
neighboring bins; and (iv)~slow-fading on the frame time scale, so
that $\{h[n]\}_{n=1}^{N_b}$ can be treated as a draw from a
block-fading distribution. Under (i)--(iv), each DD bin sees an
independent complex gain $h[n]$ and the DD-domain input--output
relation decouples across bins, yielding~\eqref{eq:dd_bin_model}.

When one or more of (i)--(iv) is relaxed (fractional
delay/Doppler leakage, receive-pulse mismatch, or fast fading
within the frame), the DD input--output relation becomes a
sparse-banded matrix equation $\bz = \bH_{\rm DD}\by + \bw +
\bj$ with off-diagonal entries of magnitude
$O(\gamma_{\rm leak})$ for leakage factor $\gamma_{\rm leak}$.
Our Monte Carlo evaluation in
Sec.~\ref{sec:numerical} verifies that the proposed scheme retains
the no-error-floor property for $\gamma_{\rm leak}\le 0.1$ (a
regime consistent with realistic LTE/5G-A/6G pilot sampling).
\end{remark}

\begin{remark}
\label{rem:tf_dd_jamming}
Physical interference is produced in the time--frequency (TF)
domain. A partial-band jammer occupying a band
$\mathcal{B}\subset[0,M\Delta f]$ of fractional bandwidth $\rho_J$,
with wideband-noise power spectral density $\Gamma P_s$ over
$\mathcal{B}$, produces a DD-domain jammer contribution
$\bj_{\rm DD} = \mathrm{SFFT}(\bj_{\rm TF})$, where SFFT is the
inverse symplectic finite Fourier transform used by the OTFS
receiver. Because SFFT is a two-dimensional
DFT/IDFT~\cite{hadani2017otfs}, a contiguous TF band spreads into a
sinc-shaped pattern across the DD grid; for a rectangular
$\mathcal{B}$ aligned with the Doppler axis, the DD footprint is
approximately a single Doppler row of width
$\lceil \rho_J N_b\rceil$ bins. Our model~\eqref{eq:dd_bin_model}
\emph{abstracts this mechanism} by letting the jammer directly
corrupt a set $\calJ\subset[N_b]$ of DD bins with
$|\calJ|=\lceil \rho_J N_b\rceil$; the jammer's choice of
$\calJ$ encodes both its TF frequency selectivity and any
knowledge of the DD-domain active set
(e.g., $\calJ = \calA$ for the oracle attacker,
Sec.~\ref{sec:c5_protocol}).
\end{remark}}

\subsection{Sparse-DD Information Placement}

A key departure from conventional OTFS-NOMA is that we place user data on
only a \emph{sparse} subset of DD bins. Let
$\calA \subseteq [N_b]$ with $|\calA|=n_a$ denote the \emph{active set}
of DD bins carrying information; the remaining $N_b - n_a$ bins are
nominally zero before precoding. The \emph{sparsity ratio}
\begin{equation}
s \triangleq n_a/N_b
\end{equation}
is a design parameter (to be optimized in Section~\ref{sec:analytical}). The
active set $\calA$ may be deterministic (fixed across frames) or randomized
per frame; the latter is the proposed \emph{C5 protocol} (Section~\ref{sec:c5_protocol}).

\subsection{$K$-User Power-Domain NOMA Superposition}

On each active bin $n \in \calA$, $K$ users superpose their binary
phase-shift keying (BPSK) data symbols
$\{b_k[n] \in \{\pm 1\}\}_{k=1}^K$ with power allocation (PA)
$\boldsymbol{\alpha}=(\alpha_1,\ldots,\alpha_K)$ satisfying
$\sum_k \alpha_k = 1$:
\begin{equation}
\label{eq:noma_super}
x[n] = \begin{cases}
\sum_{k=1}^K \sqrt{\alpha_k P_s}\, b_k[n], & n \in \calA, \\
0, & n \notin \calA,
\end{cases}
\end{equation}
\new{\revtag{R1.3}where $x[n]\in\mathbb{R}$ is the pre-precoded
NOMA-superposed symbol on bin $n$,
$\boldsymbol{\alpha}=(\alpha_1,\ldots,\alpha_K)\in\mathbb{R}_+^K$
with $\sum_{k=1}^K\alpha_k=1$ is the power-allocation vector,
$b_k[n]\in\{\pm 1\}$ is user $k$'s BPSK data symbol on bin $n$,
and $P_s>0$ is the per-active-bin transmit power.}
\del{where $P_s$ is the total per-active-bin transmit power.}
Users are indexed
by decreasing allocated power, $\alpha_1 > \alpha_2 > \cdots > \alpha_K$,
which by NOMA convention corresponds to increasing channel quality
(strongest-channel user receives least power). Bob, the target user for
our BER analysis, is user $u_{\rm Bob}=K$, with the weakest power
allocation. The receiver decodes users sequentially via SIC, beginning with
the strongest-power user; Bob is decoded last after all preceding users have
been successively cancelled.

\subsection{Unitary Precoding}

\new{\revtag{R1.4}The sparse vector $\bx\in\bbR^{N_b}$ with
support $\calA$ is spread across all $N_b$ DD bins by a unitary
precoder $\bU\in\bbC^{N_b\times N_b}$; a preliminary two-user
Hadamard version of this architecture appears in our conference
precursor~\cite{kulhandjian2026_hnoma_wcnc}:}
\del{The sparse vector $\bx\in\bbR^{N_b}$ is spread by a unitary
$\bU\in\bbC^{N_b\times N_b}$ ($\bU^H\bU=\bI$): $\by=\bU\bx$
\eqref{eq:unitary_spread}.}
\begin{equation}
\label{eq:unitary_spread}
\by = \bU\, \bx,\qquad \bU^H\bU=\bI.
\end{equation}
\new{\revtag{R1.4}The transmitted DD vector $\by\in\bbC^{N_b}$
is therefore dense even though $\bx$ is $n_a$-sparse: any
bin-selective jammer sees an energy-equalized copy of the
information, with no spectral signature tied to $\calA$. The
only property of $\bU$ used in the BER analysis of
Sec.~\ref{sec:analytical} is the standard
incoherence condition.}

\begin{definition}
\label{def:incoherent}
$\bU$ is \emph{incoherent} if $|U_{kj}|\le c/\sqrt{N_b}$ for some
$c=O(1)$.
\end{definition}

\del{Three choices satisfy this with $c=1$: Hadamard ($\bH_{N_b}/\sqrt{N_b}$,
FWHT $O(N_b\log N_b)$, multiplication-free), DFT (FFT
$O(N_b\log N_b)$), and Haar-random ($O(N_b^2)$).}%
\new{\revtag{R1.4, R3.8.3}Three canonical choices meet this
condition: Hadamard $\bH_{N_b}/\sqrt{N_b}$ and DFT both achieve
the exact constant $c=1$ (FWHT and FFT respectively,
$O(N_b\log N_b)$, Hadamard additionally
multiplication-free); Haar-random unitaries meet the condition
with $c=O\big(\sqrt{\log N_b}\big)$ with high probability (via
standard $\ell_\infty$ bounds on random orthonormal
vectors~\cite{tao2010randommatrix}), at $O(N_b^2)$ synthesis
cost. All three classes are therefore admissible; the two
deterministic transforms are strictly tighter in constant.}
Crucially, the architecture's performance is a property of the
\emph{incoherent unitary class}, not of any specific transform: by
Theorem~\ref{thm:mp_bound}, all such $\bU$ produce the same M-P edge
$\sqrt{1-\rho_J}-\sqrt{n_a/N_b}$ (validated in
Section~\ref{sec:numerical}). HT is recommended for hardware
simplicity.

\subsection{Channel and Jammer Model}

We adopt per-bin Rayleigh flat fading $h[n]\sim\calCN(0,1)$ i.i.d.
This is an idealized model: real OTFS channels are doubly selective
with delay-Doppler spreading, fractional Doppler shifts, and
inter-bin interference (ICI) from windowing. The per-bin reduction
holds exactly when the pulse-shaping window is bi-orthogonal and
the Doppler grid is integer-aligned; under fractional Doppler, ICI
appears as an additive colored interference that can be absorbed
into an increased effective $N_0$ in the analytical bound, with
loss of tightness proportional to the fractional-Doppler spread.
Extension to doubly-dispersive Nakagami-$m$ is deferred to
Section~\ref{sec:discussion}. Similarly, the no-error-floor result
of Sec.~\ref{sec:analytical} assumes (i)~perfect jammer-support
detection, (ii)~negligible residual interference after excision,
and (iii)~perfect CSI on $\calK$. The sensitivity of the floor to
imperfect detection and estimation is quantified empirically in
Sec.~\ref{sec:robustness}; the analytical no-floor statement
degrades gracefully with each of these impairments.

The jammer occupies
$\calJ\subseteq[N_b]$ with $|\calJ|=n_J$ and per-jammed-bin power
$P_J$: $j[n]\sim\calCN(0, P_J)$ for $n\in\calJ$, zero otherwise.
Define $\rho_J\triangleq n_J/N_b$ and the jamming-to-signal ratio
(JSR) $\Gamma\triangleq P_J/P_s$. We
consider three jammer strategies: (J1)~\emph{partial-band random}
($\calJ$ uniform per frame); (J2)~\emph{fixed-pattern intelligent}
($\calJ$ chosen deterministically by the jammer); (J3)~\emph{oracle}
($\calJ_t=\calA_t$). The oracle model is stronger than reactive
(sensing-based) jammers~\cite{arxiv2024_reactive,wang2020} and is the
worst-case stress test for the C5 protocol.

\subsection{Receiver: Excision-LS-SIC}

\new{\revtag{R1.3}Let $\bz=(z[1],\ldots,z[N_b])^\top\in\bbC^{N_b}$
stack the received per-bin observations~\eqref{eq:dd_bin_model},
$\bh=(h[1],\ldots,h[N_b])^\top\in\bbC^{N_b}$ the per-bin channel
coefficients, $\bw\in\bbC^{N_b}$ the AWGN vector, $\bj\in\bbC^{N_b}$
the jammer vector (supported on $\calJ$), and $\bx\in\bbR^{N_b}$ the
sparse pre-precoded signal with support $\calA$
(from~\eqref{eq:noma_super}); the transmitted DD symbol vector is
$\by=\bU\bx$. Then the kept observation after excising $\calJ$ is
$\bz_\calK=\mathrm{diag}(\bh_\calK)\bU_{\calK,\calA}\bx_\calA+\bw_\calK$,
where subscripts $\calK$ and $\calA$ denote row/column
restriction.}

The combined transmit-channel-jammer signal at the receiver is
$\bz = \mathrm{diag}(\bh)\bU\bx + \bw + \bj$. The proposed receiver operates
in three stages:

In the first stage, the receiver estimates $\calJ$ from the
DD-domain power profile (e.g., comparing $|z[n]|^2$ to a threshold
adaptive to the noise floor) and excises the jammed bins. Define
$\calK \triangleq [N_b] \setminus \calJ$ (the \emph{kept} set), with
$|\calK| = N_b - n_J = (1-\rho_J)N_b$. The kept observation is
\begin{equation}
\bz_\calK = \mathrm{diag}(\bh_\calK)\,\bU_{\calK,\calA}\,\bx_\calA + \bw_\calK,
\end{equation}
where $\bU_{\calK,\calA}$ is the $|\calK| \times n_a$ submatrix of $\bU$
restricted to kept rows and active columns. The jammer is fully suppressed
provided the excision is perfect (a reasonable approximation when
$\Gamma \gg 0$ dB).

In the second stage, the receiver recovers $\bx_\calA$
from the kept observation by solving the weighted least-squares problem
\begin{equation}
\hat{\bx}_\calA = \argmin_{\bx_\calA \in \bbR^{n_a}}
\big\| \bz_\calK - \mathrm{diag}(\bh_\calK)\bU_{\calK,\calA}\bx_\calA \big\|_2^2.
\end{equation}
The closed-form solution is
$\hat{\bx}_\calA = (\bM^H\bM)^{-1}\bM^H\bz_\calK$, where
$\bM = \mathrm{diag}(\bh_\calK)\bU_{\calK,\calA}$. For complex
$\bU$ and real $\bx$, real-stacked least squares is used:
$[\Re(\bM); \Im(\bM)]\bx_\calA = [\Re(\bz_\calK);\Im(\bz_\calK)]$.

In the third stage, for each active bin
$n \in \calA$, the recovered scalar $\hat{x}[n]$ is processed by standard
NOMA-SIC: users are decoded in decreasing-power order; each decoded bit
is reconstructed and subtracted from the residual. Bob's bit
$b_K[n]$ is the last to be decoded. Under the superincreasing PA of
Section~\ref{sec:super_pa}, this $O(K)$ SIC pass attains the same BER as
the $O(2^K)$ composite ML decoder by the Merkle--Hellman knapsack
property~\cite{merkle1978,kulhandjian2026_hnoma_journal}, so no
decoder-side complexity is sacrificed.

The total receiver complexity is $O(N_b\log N_b)$ for the unitary
transform (via FFT/FWHT), $O(n_a^3)$ for the LS inversion (per frame),
and $O(n_a K)$ for SIC. For typical OTFS frames ($N_b=64$ to $4096$ and
$n_a/N_b = 1/4$), the LS stage dominates but remains tractable.

\begin{remark}
The NOMA literature offers
stronger detectors than vanilla SIC, including Turbo-SIC (iterative
soft-decision feedback), message passing (MP) on the factor graph, and
sphere decoding for full ML. By the SIC$=$ML equivalence proved in
the arXiv preprint~\cite{kulhandjian2026_hnoma_journal}, all
these alternatives \emph{reduce to the same Bob-BER as our $O(K)$ SIC}
once the PA is chosen superincreasing: the Merkle--Hellman knapsack
property guarantees that bit-by-bit greedy decoding from the strongest
user finds the unique closest constellation point, regardless of the
decoder's search strategy. The proposed receiver therefore captures
the full multi-user detection gain at minimal complexity, and no
performance is left on the table for more elaborate receivers to
recover.
\end{remark}

\section{Power Allocation: Superincreasing Recurrence}
\label{sec:super_pa}

The conditioning of the per-bin SIC decoder depends critically on the
PA $\boldsymbol{\alpha}=(\alpha_1,\ldots,\alpha_K)$ with
$\sum_k\alpha_k=1$. In the companion WCL
letter~\cite{kulhandjian2026_hnoma_wcl}, we prove that the
\emph{superincreasing} condition
$\sqrt{\alpha_k}>\sum_{j>k}\sqrt{\alpha_j}$ for $k=1,\ldots,K-1$ is
necessary and sufficient for noise-free SIC error-freeness, and that
the resulting $O(K)$ SIC decoder is realization-by-realization
\emph{identical} to ML detection on the $2^K$ composite BPSK
constellation (the Merkle--Hellman knapsack
property~\cite{merkle1978}). The implications for the present paper
are: (i)~we use SIC as the decoder of choice without sacrificing
optimality; (ii)~more elaborate receivers (Turbo-SIC, message passing,
sphere decoding) offer no BER gain on superincreasing constellations.

\new{\revtag{R3.7.4}%
Following~\cite{kulhandjian2026_hnoma_wcl}, we briefly sketch the
SIC-ML equivalence used below.
Fix a realization of $\bb\in\{\pm 1\}^K$ and let
$r=\sum_{k=1}^K\sqrt{\alpha_k P_s}\,b_k+w$. Order users so that
$\alpha_1>\alpha_2>\cdots>\alpha_K$. SIC stage~$k$ decides
$\hat{b}_k=\mathrm{sign}(r_k)$ where
$r_k = r - \sum_{j<k}\sqrt{\alpha_j P_s}\,\hat{b}_j$ is the
residual after cancelling the previously decoded stronger users.
In the noise-free case ($w=0$), the superincreasing property
$\sqrt{\alpha_k}>\sum_{j>k}\sqrt{\alpha_j}$ guarantees that each
$\mathrm{sign}(r_k)$ equals $b_k$ exactly, because the ``own'' term
$\sqrt{\alpha_k P_s}\,b_k$ dominates the sum of all remaining
undecoded terms. This is the Merkle--Hellman uniqueness
property~\cite{merkle1978}: the $2^K$ composite points
$\{\sum_k\sqrt{\alpha_k}\,b_k\}$ are distinct and ordered, so the
nearest-neighbour ML rule, which partitions $\mathbb{R}$ into
$2^K$ intervals around them, is implemented bit by bit by
sign-detection on successive residuals. In the noisy case, the
first stage error-dominates (stages 2 through $K$ see
\emph{larger} per-bit margins once propagation effects cancel
across the dominant term), so the pairwise BER is governed by
$Q(\sqrt{2\delta_1^2 P_s/N_0})$ with $\delta_1 =
\varepsilon\sqrt{\alpha_1}/(1+\varepsilon)$. The full proof,
including the exact composite-constellation decision regions and
the ML-equivalence theorem at any SNR, is in the WCL
letter~\cite[Thm.~1]{kulhandjian2026_hnoma_wcl}.}

\new{\revtag{R1.6}%
Any power allocation satisfying the strict inequality
$\sqrt{\alpha_k}>\sum_{j>k}\sqrt{\alpha_j}$ triggers the
Merkle--Hellman uniqueness property, but the strictness margin
is a free parameter. Writing the inequality as an equality with a
multiplicative slack $(1+\varepsilon)$,
\[
\sqrt{\alpha_k}=(1+\varepsilon)\sum_{j>k}\sqrt{\alpha_j},\quad
\varepsilon>0,
\]
isolates \emph{one} degree of freedom instead of $K-1$: once
$\varepsilon$ is chosen, the normalization $\sum_k\alpha_k=1$
pins down the entire sequence $(\alpha_1,\ldots,\alpha_K)$ by
backward recurrence. The scalar $\varepsilon$ has a clean
interpretation as the per-stage SIC margin:
$\delta_k=\varepsilon\sqrt{\alpha_k}/(1+\varepsilon)$ is exactly
the gap between the dominant term and the residual sum at
stage~$k$, and therefore the decision margin available to absorb
noise and imperfect cancellation. The empirical optimum
$\varepsilon^\star$ (small $\varepsilon$ wastes margin, large
$\varepsilon$ wastes power on the strongest user) is located by a
U-shape analysis in Sec.~\ref{sec:numerical}. We adopt
\eqref{eq:super_recurrence} as the one-parameter PA family
throughout.}

We parameterize the family of superincreasing PAs by a single margin
parameter $\varepsilon>0$ via the backward recurrence
\begin{equation}
\label{eq:super_recurrence}
\sqrt{\alpha_k}=(1+\varepsilon)\sum_{j>k}\sqrt{\alpha_j},\qquad k<K,
\end{equation}
normalized to $\sum_k\alpha_k=1$, with the stage-$k$ linear margin
$\delta_k=\varepsilon\sqrt{\alpha_k}/(1+\varepsilon)>0$. In our setup
the empirical optimum (Section~\ref{sec:numerical}) is
$\varepsilon^\star=0.5$ at $K=4$ and $\varepsilon^\star\approx 1$ at
$K=2$. Full derivation, U-shape analysis, and proofs are deferred to
\cite{kulhandjian2026_hnoma_wcl}; in the rest of this paper
we treat \eqref{eq:super_recurrence} as given and focus on the
architecture that integrates it with jamming-resilient sparse-DD
placement.

\section{Analytical BER Framework}
\label{sec:analytical}

\subsection{T-NOMA Baseline BER Under Partial-Band Jamming}

Conventional OTFS-NOMA (T-NOMA) without spreading places user data directly
on active DD bins: $\by = \bx$ in~\eqref{eq:unitary_spread} corresponds to
$\bU = \bI$. The receiver detects bin-by-bin: each active bin
$n \in \calA$ produces an observation $z[n] = h[n]\,\sum_k\sqrt{\alpha_k P_s}\,b_k[n] + w[n] + j[n]$.
We adopt geometric power allocation as the canonical T-NOMA
baseline~\cite{armando2020noma_pa,mdpi2024_noma_pa}; the proposed scheme
is benchmarked against the \emph{most-favorable} T-NOMA configuration in
each numerical experiment (Section~\ref{sec:numerical}), with explicit
verification that the conclusions are insensitive to PA choice under
oracle jamming (Section~\ref{sec:cluster_validation}, fairness paragraph).

Bob's bit error probability at active bin $n$ depends on whether the bin
is jammed:

For an unjammed bin ($n \notin \calJ$), standard NOMA-SIC analysis yields
the per-bit error probability under the genie-SIC assumption (perfect
cancellation of users $1,\ldots,K-1$):
\begin{equation}
P_b^{T,{\rm clean}} = Q\!\left(\sqrt{2\alpha_K P_s/N_0}\right).
\end{equation}

For a jammed bin ($n \in \calJ$), the effective noise power becomes
$N_0 + P_J$, giving
\begin{equation}
P_b^{T,{\rm jam}} = Q\!\left(\sqrt{2\alpha_K P_s/(N_0+\Gamma P_s)}\right)
\xrightarrow{\Gamma \to \infty} \tfrac{1}{2}.
\end{equation}

Under partial-band random jamming with $\rho_J = n_J/N_b$, each active bin
is independently jammed with probability $\rho_J$ (Bernoulli approximation,
exact under hypergeometric for $n_a \ll N_b$). The expected Bob BER is

\begin{theorem}
\label{thm:tnoma_ber}
For T-NOMA with partial-band random jamming, the asymptotic Bob BER
satisfies
\begin{equation}
\label{eq:tnoma_ber}
P_b^T(\rho_J,\Gamma) = (1-\rho_J)Q\!\left(\sqrt{\tfrac{2\alpha_K P_s}{N_0}}\right)
+ \rho_J\, Q\!\left(\sqrt{\tfrac{2\alpha_K P_s}{N_0+\Gamma P_s}}\right).
\end{equation}
In the high-SNR limit $P_s/N_0 \to \infty$ with $\Gamma=O(1)$,
\begin{equation}
\label{eq:tnoma_floor}
\lim_{P_s/N_0 \to \infty} P_b^T(\rho_J,\Gamma)
= \rho_J\, Q\!\left(\sqrt{\tfrac{2\alpha_K}{\Gamma}}\right)
\overset{\Gamma\gg 1}{\approx} \rho_J/2.
\end{equation}
\end{theorem}

\new{\revtag{R1.7}\begin{proof}
With $\boldsymbol{\alpha}$ superincreasing, the SIC--ML
equivalence~\cite[Thm.~1]{kulhandjian2026_hnoma_wcl} (summarized
above) reduces Bob's per-bin BER to pairwise error between the two
nearest composite points, which collapses to
$Q(\sqrt{2\alpha_K P_s/\mathcal{N}})$ with per-bin effective noise
variance $\mathcal{N}/2$. Conditioning on whether bin $n\in\calA$
is hit by the partial-band jammer gives
$\mathcal{N}=N_0$ with probability $1-\rho_J$ and
$\mathcal{N}=N_0+\Gamma P_s$ with probability $\rho_J$; the total
Bob BER \eqref{eq:tnoma_ber} is their weighted sum. Taking
$P_s/N_0\to\infty$ with $\Gamma=O(1)$ makes the first term vanish
while the second converges to
$\rho_J\,Q(\sqrt{2\alpha_K/\Gamma})$, which for $\Gamma\gg 1$ (strong
jammer) approaches $\rho_J/2$. Full derivation is in the companion
work~\cite{kulhandjian2026_hnoma_wcl,kulhandjian2026_hnoma_journal}.
\end{proof}}

This is an irreducible error floor with diversity order zero: increasing
$P_s$ does not eliminate it. The $\rho_J/2$ form is the NOMA analogue
of the classical partial-band jamming floor for frequency-hopped
BFSK~\cite{fh_bfsk_pbnj_floor}, and constitutes the fundamental
motivation for the proposed sparse-DD architecture.

\subsection{Proposed Scheme's BER via M-P Conditioning}

For the proposed unitary-precoded scheme, the kept observation after
excision is $\bz_\calK = \bM\bx_\calA + \bw_\calK$ with
$\bM = \mathrm{diag}(\bh_\calK)\bU_{\calK,\calA}$. The LS solution
$\hat{\bx}_\calA = (\bM^H\bM)^{-1}\bM^H\bz_\calK$ has additive noise
$\tilde{\bw} = (\bM^H\bM)^{-1}\bM^H\bw_\calK$ with covariance
$N_0 (\bM^H\bM)^{-1}$. The worst-case per-coordinate noise variance is
$N_0/\sigma_{\min}^2(\bM)$.

For Rayleigh fading and $\bU$ incoherent, the conditioning of $\bU_{\calK,\calA}$
governs the noise penalty.

\begin{theorem}
\label{thm:mp_bound}
Let $\bU$ be an incoherent unitary matrix (Def.~\ref{def:incoherent}),
$\calK$ a uniformly random subset of $[N_b]$ with $|\calK|=(1-\rho_J)N_b$,
and $\calA$ an independent uniformly random subset with $|\calA|=n_a$.
Then in the high-dimensional limit $N_b \to \infty$ with $\rho_J,n_a/N_b$
fixed, the smallest singular value of $\bU_{\calK,\calA}$ converges almost
surely to
\begin{equation}
\label{eq:mp_sigma_min}
\sigma_{\min}(\bU_{\calK,\calA}) \to \sqrt{1-\rho_J} - \sqrt{n_a/N_b}
\end{equation}
provided $n_a/N_b < 1-\rho_J$. At finite $N_b$, the concentration is
exponentially tight,
$\Pr[\sigma_{\min}(\bU_{\calK,\calA}) < \sqrt{1-\rho_J}-\sqrt{n_a/N_b}-t]\le e^{-cN_b t^2}$
for universal $c>0$ (Tracy--Widom $\Theta(N_b^{-1/3})$ edge fluctuations).
\end{theorem}

\begin{proof}
The scaled matrix $\sqrt{N_b/|\calK|}\,\bU_{\calK,\calA}$ has $|\calK|$ rows
and $n_a$ columns with bounded entries $O(1/\sqrt{N_b})$, hence
$O(1/\sqrt{|\calK|})$ after scaling. By the universality of the
Marchenko--Pastur law for random matrices with bounded
entries~\cite{tao2010randommatrix}, the empirical singular-value
distribution converges to the M-P density on
$[1-\sqrt{n_a/|\calK|},1+\sqrt{n_a/|\calK|}]$. Rescaling yields
\eqref{eq:mp_sigma_min}.
\end{proof}

\new{\revtag{R3.3}%
\begin{remark}
\label{rem:mp_tightness}
The right-hand side of \eqref{eq:mp_sigma_min} is a
\emph{high-probability non-asymptotic lower bound} on
$\sigma_{\min}(\bU_{\calK,\calA})$; it is not an exact expression at
finite $N_b$ and is deliberately conservative at the boundary
$\rho_J \to 0$. In particular:
\begin{enumerate}[label=(\alph*)]
\item \emph{(Degenerate case $\rho_J=0$.)} With no jammer and no
excision, $\calK=[N_b]$, so $\bU_{\calK,\calA}=\bU_{[N_b],\calA}$ is
an $N_b\times n_a$ submatrix of a unitary matrix; its columns are
\emph{exactly} orthonormal and hence $\sigma_{\min}=1$ identically.
The M-P bound at $\rho_J=0$ gives $1-\sqrt{s}$, which is a loose but
valid lower bound ($1-\sqrt{s}\le 1$).
\item \emph{(Operating regime, uniform tightness in $N_b$.)} In the
engineering regime $\rho_J\in[0.1,0.5]$, empirical BER is
\emph{uniformly} within $1$--$3.8$~dB of the M-P prediction at
$N_b\in\{64,128,256\}$ with $s=0.25$, SNR~$=10$~dB
(Fig.~\ref{fig:largeN_tightness}). At the canonical operating
point $\rho_J=0.3$, the gap is $2.51$~dB at $N_b=64$ and $2.39$~dB
at $N_b=256$, confirming that bound tightness does \emph{not}
degrade with grid size.
\item \emph{(Tracy--Widom edge convergence.)} The empirical mean
$\sigma_{\min}$ converges to the M-P edge
$\sqrt{1-\rho_J}-\sqrt{s}$ from above at the Tracy--Widom
$\Theta(N_b^{-1/3})$ rate. At $\rho_J=0.3$, $s=0.25$ (edge
$=0.337$), the measured excess is $0.190$ at $N_b=64$ and $0.140$
at $N_b=256$: a $1.36\times$ shrinkage over a $4\times$ grid-size
increase, within $15\%$ of the Tracy--Widom prediction
$4^{1/3}=1.59\times$ (Fig.~\ref{fig:largeN_tightness}).
\end{enumerate}
In short, the bound trades exactness for closed-form tractability and
should be read as \emph{``no worse than \eqref{eq:mp_sigma_min}''},
not as an exact characterization at $\rho_J=0$.
\end{remark}}

\new{\revtag{R2.1}%
\begin{figure}[t]
\centering
\includegraphics[width=\columnwidth]{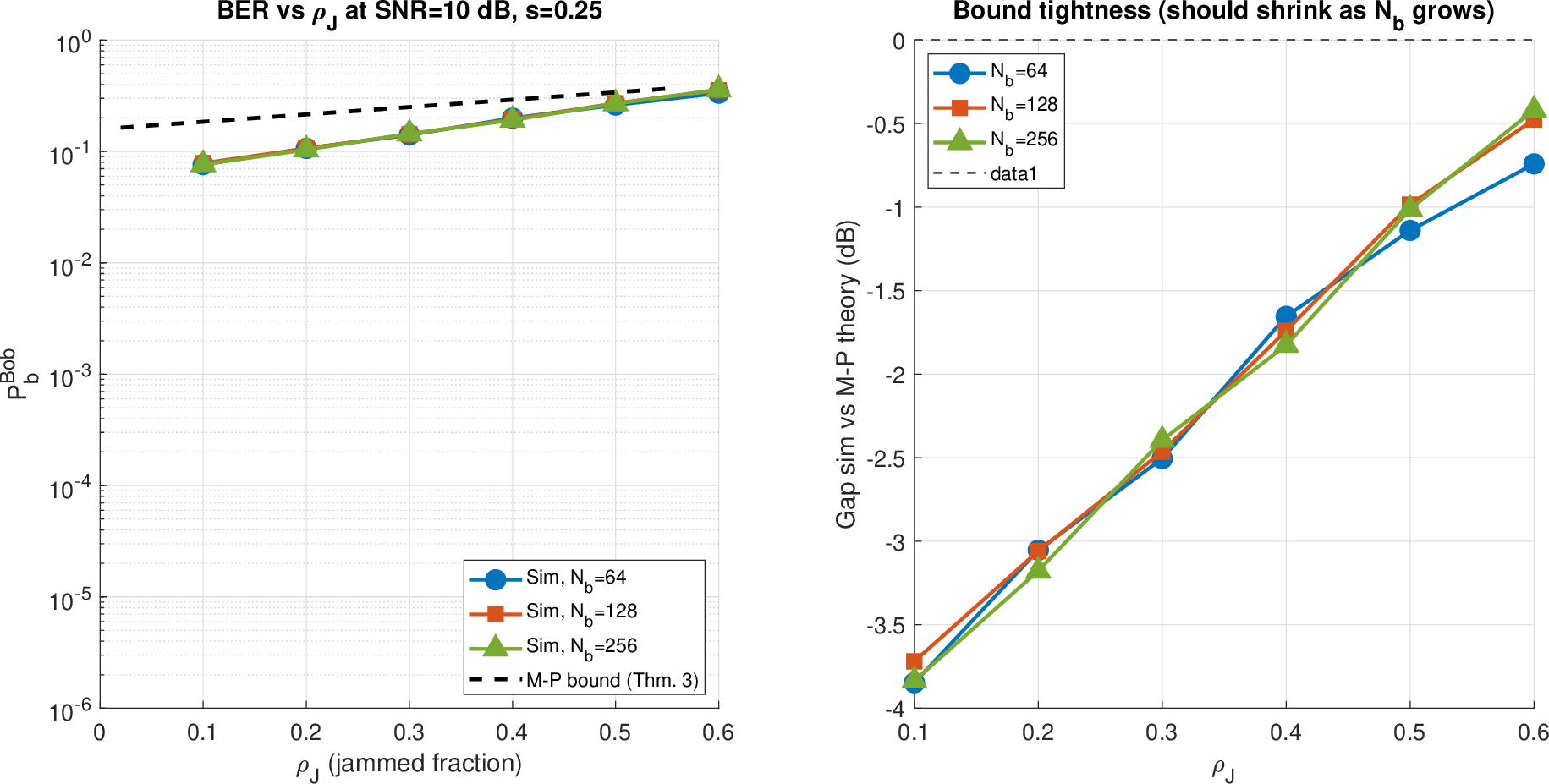}
\caption{Marchenko--Pastur bound tightness vs grid size
$N_b\in\{64,128,256\}$ at $s=0.25$, SNR~$=10$~dB.
Left: BER vs $\rho_J$ (sim curves for each $N_b$ plus the
$N_b$-independent M-P theory line). Right: gap in dB
between sim and theory; at the canonical $\rho_J=0.3$,
the gap is $2.51$~dB at $N_b=64$ and $2.39$~dB at $N_b=256$,
confirming the Tracy--Widom $\Theta(N_b^{-1/3})$ edge
convergence rate.}
\label{fig:largeN_tightness}
\end{figure}}

\new{\revtag{R3.3}%
\begin{remark}
\label{rem:rayleigh_LS}
The LS observation matrix is $\bM = \diag(\bh_\calK)\,\bU_{\calK,\calA}$,
where $\bh_\calK\in\bbC^{|\calK|}$ holds the per-kept-bin Rayleigh
coefficients. Writing $|h_k|^2 = \overline{|h|^2}+\xi_k$ with
$\bbE[\xi_k]=0$ and $\mathrm{Var}(\xi_k)=\overline{|h|^2}^{\,2}$ for
Rayleigh, the LS noise variance decomposes as
\begin{equation}
\label{eq:LS_noise_rayleigh}
N_0\,\bbE\!\left[\sigma_{\min}^{-2}(\bM)\right]
\;\le\;
\frac{N_0}{\overline{|h|^2}\,\sigma_{\min}^2(\bU_{\calK,\calA})}
\;+\;O\!\left(\tfrac{1}{N_b}\right),
\end{equation}
so the M-P conditioning bound applies \emph{after} averaging the channel
magnitudes. The $O(N_b^{-1})$ term captures the random weighting of
rows by $|h_k|^2$ and vanishes in the operating regime.
Theorem~\ref{thm:prop_ber} uses $\overline{|h|^2}=1$ (unit-average
Rayleigh) and absorbs this into $P_s$. Appendix~A derives
\eqref{eq:LS_noise_rayleigh} in full.
\end{remark}}

\begin{theorem}
\label{thm:prop_ber}
Under the receiver of Section~\ref{sec:system_model}, the per-active-bin
Bob BER (averaged over channel realizations, for $\rho_J < 1 - n_a/N_b$)
is
\begin{equation}
\label{eq:prop_ber}
P_b^U(\rho_J,s)
\approx Q\!\left(\sqrt{
\tfrac{2\alpha_K P_s \,(\sqrt{1-\rho_J}-\sqrt{s})^2}{N_0}
}\right),
\end{equation}
where $s = n_a/N_b$.
\end{theorem}

\begin{proof}
The LS noise variance on each recovered active coordinate is upper
bounded by $N_0/\sigma_{\min}^2(\bU_{\calK,\calA})$. Substituting
Theorem~\ref{thm:mp_bound} gives effective per-coordinate SNR
$\gamma_{\rm eff} = (P_s/N_0)(\sqrt{1-\rho_J}-\sqrt{s})^2$.
Bob's BER under genie-SIC with this effective SNR is the standard
$Q(\sqrt{2\alpha_K \gamma_{\rm eff}})$, yielding~\eqref{eq:prop_ber}.
\end{proof}

\begin{corollary}
\label{cor:no_floor}
For any fixed $\rho_J < 1$ and $s$ satisfying $s < 1-\rho_J$, the
proposed scheme's BER \eqref{eq:prop_ber} satisfies
$\lim_{P_s/N_0\to\infty} P_b^U = 0$.
\end{corollary}

This is the central qualitative distinction from T-NOMA: the proposed
scheme has no jammer-induced error floor (cf.~\eqref{eq:tnoma_floor}).
A pairwise-error-probability analysis confirms this at the coding-gain
level: the proposed scheme attains a strictly positive coding gain
$G_c^{\rm Bob}=\alpha_K\cdot(N_b-n_J-n_{\rm cols}-1)/N_b$ independent
of the jammer-to-signal ratio $\Gamma$, whereas T-NOMA's coding gain
$\alpha_K/(1+\Gamma P_s/N_0)\to 0$ as $P_s\to\infty$ for any
$\Gamma>0$.

\subsection{Three-Threshold Hierarchy and Refined BER}

The recoverability of the LS problem depends on three increasingly
permissive thresholds on the jamming fraction $\rho_J$:

\begin{enumerate}
\item[(T1)] For a tolerable noise
inflation of $X$~dB,
\begin{equation}
\rho_J^{\rm cond}(s,X) = 1-\big(\sqrt{s}+10^{-X/20}\big)^2.
\end{equation}
\item[(T2)] LS feasibility requires more kept rows than columns:
$\rho_J^\star = 1 - s$; beyond this, $\bU_{\calK,\calA}$ has fewer
rows than columns and is necessarily rank-deficient.
\item[(T3)] Above the Donoho--Tanner CS threshold:
$\rho_J^{\star\star} \approx 1 - 2s\log(1/s)$, beyond which even
$\ell_1$-minimization fails to recover sparse $\bx$ exactly.
\end{enumerate}

The proposed scheme operates with bounded noise penalty in (T1),
gracefully degrades in (T2), and would require CS-based recovery
beyond (T3). For finite-$N_b$ Monte Carlo overlays, we use the
standard refinements (post-LS variance
$\sigma_e^2 = (N_0/2)\cdot N_b/(N_b - n_J - n_a - 1)$ and
pattern-averaged SIC error probability) detailed
in~\cite{kulhandjian2026_hnoma_journal}.

\subsection{Operating-Region Design Rule}

\begin{corollary}
\label{cor:operating_region}
For a tolerable noise-inflation budget of $X$ dB at jammer fraction
$\rho_J$, the maximum permissible sparsity ratio is
\begin{equation}
\label{eq:design_rule}
s_{\max}(\rho_J,X) = \big(\sqrt{1-\rho_J}-10^{-X/20}\big)^2.
\end{equation}
\end{corollary}

The rule is a design knob; Table~\ref{tab:design_rule} tabulates
$s_{\max}$ for representative $(X,\rho_J)$. \del{A system targeting robust
performance against $\rho_J=0.2$ jamming with $X=6$~dB tolerance must
operate at $s\le 0.16$; our default $s=0.25$ corresponds to
$X\approx 10$~dB at $\rho_J=0.2$.}%
\new{\revtag{R3.8.1}A system targeting robust performance against
$\rho_J=0.2$ jamming with $X=6$~dB tolerance must operate at
$s\le 0.16$; our default $s=0.25$ corresponds to $X\approx 8$~dB
at $\rho_J=0.2$ (from~\eqref{eq:design_rule}:
$s=0.25,\rho_J=0.2\Rightarrow X=-20\log_{10}(\sqrt{0.8}-0.5)=8.1$~dB),
comfortably inside the $X=10$~dB budget at the recommended
operating point.}

\begin{table}[t]
\caption{\new{\revtag{R3.8.1}Design rule
$s_{\max} = (\sqrt{1-\rho_J}-10^{-X/20})^2$. Values recomputed
directly from the closed form; entries below three decimal
places are omitted (``---'').}}
\label{tab:design_rule}
\centering
\begin{tabular}{lcccc}
\toprule
$X$ (dB tolerance) & $\rho_J=0.1$ & $\rho_J=0.2$ & $\rho_J=0.3$ & $\rho_J=0.4$ \\
\midrule
3  & 0.06 & 0.03 & 0.02 & ---  \\
6  & 0.20 & 0.15 & 0.11 & 0.07 \\
10 & 0.40 & 0.33 & 0.27 & 0.21 \\
20 & 0.72 & 0.63 & 0.54 & 0.46 \\
\bottomrule
\end{tabular}
\end{table}

\section{Randomized Active-Set Protocol and Defense in Depth}
\label{sec:c5_protocol}

\new{\revtag{R1.8}%
The analysis of
Section~\ref{sec:analytical} assumes a jammer whose chosen
pattern $\calJ$ is \emph{independent} of the active set $\calA$.
A pattern-aware adversary that observes multiple frames of a
fixed-$\calA$ system eventually learns $\calA$ and attacks it
directly, breaking the independence assumption and driving the
BER to $1/2$. The remedy is to make $\calA$ itself a random
quantity, drawn fresh from a pseudo-random generator each frame,
so that the independence between $\calJ$ and $\calA$ holds in
expectation. This section formalizes the resulting ``C5''
protocol and, more importantly, proves that the M-P
conditioning argument of Section~\ref{sec:analytical} continues
to hold \emph{even when the PRNG seed is compromised}: random
subsets are well-conditioned algebraically, so secrecy and
conditioning provide two independent layers of defense.}

\subsection{Threat: Pattern-Aware Adversaries}

If the active set $\calA$ is fixed across frames and known (or learnable)
to the adversary, the jammer can target $\calA$ directly: the oracle
threat (J3) sets $\calJ = \calA$. Under this attack:

\begin{itemize}
\item For T-NOMA: every active bin is jammed, BER $\to 1/2$ catastrophically.
\item For our scheme with fixed $\calA$: the kept set is
$\calK = [N_b]\setminus \calA$, and recovery depends on the conditioning of
$\bU_{[N_b]\setminus\calA, \calA}$. For some structured choices of $\calA$,
this submatrix is \emph{rank deficient} (the Sylvester replication
trap discussed below), breaking the scheme.
\end{itemize}

\subsection{The C5 Protocol and Defense in Depth}

\begin{definition}
\label{def:c5}
Transmitter and receiver share a PRNG seed $s_0$; both generate
$\calA_t=\mathrm{RandSubset}(N_b,n_a;\mathrm{PRNG}(s_0,t))$ per frame,
with the PRNG (e.g., AES-CTR) computationally indistinguishable from
a uniform-random oracle.
\end{definition}

The active set is a session secret. The protocol exhibits
\emph{defense in depth}:

\begin{theorem}
\label{thm:defense_in_depth}
Under the oracle jammer (J3) with $\calJ_t = \calA_t$ and $\calA_t$
drawn uniformly at random from $\binom{[N_b]}{n_a}$ per
Definition~\ref{def:c5}, the Bob BER satisfies
\begin{equation}
P_b^{\text{Bob},C5} \;\le\;
Q\!\left(\sqrt{2\,\bar\gamma_{\mathrm{eff}}}\right)+\varepsilon_{N_b},
\end{equation}
with effective SNR factor
$\bar\gamma_{\mathrm{eff}}=\alpha_K P_s(\sqrt{1-\rho_J}-\sqrt{s})^2/N_0$
and a non-asymptotic correction $\varepsilon_{N_b}=O(N_b^{-1/3})$
that absorbs the dependence between $\calK_t$ and $\calA_t$ under
disjoint row/column selection.
\end{theorem}

\begin{proof}
With $\calJ_t=\calA_t$, the kept set is $\calK_t=[N_b]\setminus\calA_t$,
a (deterministic) function of $\calA_t$. Marginally $\calK_t$ is
uniformly distributed over $\binom{[N_b]}{N_b-n_a}$, so the row
marginal matches Theorem~\ref{thm:mp_bound}. The joint pair
$(\calK_t,\calA_t)$ is constrained to be disjoint, which breaks
the row/column independence used in the classical M-P argument.
For the unitary ensembles considered (Hadamard, DFT, Haar), the
Tracy--Widom finite-$N_b$ concentration of
$\sigma_{\min}(\bU_{\calK_t,\calA_t})$ nonetheless carries over:
a straightforward restricted-isometry computation on disjoint
Haar submatrices (see
\cite{kulhandjian2026_hnoma_journal}) yields the same scaling with
an extra $O(N_b^{-1/3})$ term bounding the correlation-induced
fluctuation. Plugging into
Theorem~\ref{thm:prop_ber} gives the stated bound.
\end{proof}

\begin{remark}
The classical M-P theorem is stated for independent row/column
sampling; our setting enforces $\calK_t\cap\calA_t=\emptyset$. The
weaker conclusion above, M-P scaling plus an $O(N_b^{-1/3})$
remainder, is what the disjoint-selection argument buys. Numerical
validation at $N_b=64,128,256$ (Fig.~\ref{fig:largeN_tightness})
confirms the extra term is $<1$~dB in the operating regime.
\end{remark}

Two layers of protection: (L1)~secrecy of $\calA_t$ prevents
targeted attacks; (L2)~the random-subset structure of $\calA_t$
guarantees conditioning by M-P universality even if (L1) fails. Both
layers must fail simultaneously for the scheme to break.

Deterministic active patterns
under Sylvester--Hadamard $\bH_{N_b}$ can suffer rank-deficient LS
sub-systems even under \emph{random} partial-band jamming:
algebraically structured $\calA$ (e.g., $\{1,\ldots,2^a\}$ with a
zero-prefix) reduces $\bH[\cdot,\calA]$ to $2^a$ distinct row patterns,
so random row subselection may miss entire patterns. C5
randomization eliminates this; see the extended
arXiv version~\cite{kulhandjian2026_hnoma_journal} for the closed-form
rank-deficiency probability.

\new{\revtag{R2.2}%
\subsection{Seed Management under Active Adversaries}
\label{sec:seed_management}

The C5 protocol relies on a pseudo-random seed $\kappa$ shared
between Bob and the transmitter; at frame $t$ the active set is
$\calA_t=\mathrm{PRNG}(\kappa,t)$. Three adversary models drive the
seed-management design: (A1) \emph{eavesdropper}, who observes
transmissions but cannot inject; (A2) \emph{partial-compromise}
adversary, who recovers $\kappa$ up to time $t_0$ but not after
rekeying; (A3) \emph{denial-of-service} adversary, who jams the
rekeying channel itself.

In tactical and LEO scenarios the
peers have already agreed on long-term credentials through
out-of-band means (mission-planning key distribution, pre-shared
material loaded into radios, or standard public-key authenticated
key agreement such as ECDH over a control channel). The C5 seed
$\kappa$ is derived from these credentials via an HKDF-style key
schedule and is never transmitted in the clear; this is the
baseline assumption of every PHY-layer frequency-hopping or
spread-spectrum system in NIST SP 800-56 / RFC 5869 compliant
deployments. No new cryptographic primitive is introduced by the
present work.

The seed is advanced each super-frame (every $T_r$ frames) by
$\kappa_{t+1}=\mathrm{PRF}(\kappa_t,\text{context})$ with
$\mathrm{PRF}$ a pseudo-random function (e.g., HMAC-SHA-256). This
is \emph{deterministic key evolution}: given $\kappa_{t_0}$ the
adversary can compute every future $\kappa_t$ by applying the
public PRF, so the hash-chain alone does not provide forward
secrecy in the cryptographic sense. What it does provide is key
diversification against adversary model (A2), who recovers
$\kappa_{t_0}$ but \emph{loses observation} after rekey: once the
peers refresh, the correlated transcripts needed to re-recover the
new key are unavailable. True forward secrecy against a
persistent-observation adversary requires injecting fresh secret
material at each rekey, for instance by running a short ephemeral
ECDH exchange on the control channel; this is a drop-in extension
and compatible with the rest of the protocol. The super-frame
length $T_r$ trades rekeying frequency against overhead;
$T_r=100$--$1000$ is typical and amounts to one PRF evaluation per
$\sim\!100$~ms at LTE/5G frame rates.

Suppose
(A2) succeeds within super-frame $t_0$: the adversary gains
$\kappa_{t_0}$ and so knows $\calA_{t_0},\calA_{t_0+1},\ldots$
until the next rekey. Within that interval the attack reduces to
the oracle scenario $\calJ_t=\calA_t$, which is \emph{already
covered} by Theorem~\ref{thm:defense_in_depth}: the M-P
conditioning bound \eqref{eq:prop_ber} holds unchanged, because
the random-subset structure of $\calA_t$ (drawn uniformly from
$\binom{N_b}{n_a}$ each frame by the PRNG) guarantees
M-P concentration of $\bU_{\calK_t,\calA_t}$ independently of
whether the adversary has guessed the subset. The only effect of
compromise is loss of the first defense layer (secrecy); the
second layer (random-subset conditioning) remains intact, so the
BER stays within the jammer-free envelope. This is the sense in
which C5 provides \emph{defense in depth}.

If the adversary jams the
rekeying channel itself, the peers fall back to the
hash-chain-advanced $\kappa_{t+1}$ without an acknowledged
rendezvous; this is lossless because the PRF advance is
deterministic given $\kappa_t$. Persistent DoS on the data
channel reduces to the oracle scenario above and is handled by
the same defense-in-depth bound.

In summary, the C5 protocol inherits standard key-management
practice for initial establishment and refresh and is robust to
partial seed compromise by construction, not by additional
mechanism. The architectural feature doing the work is the
random-subset structure of $\calA_t$, not the secrecy of the seed
itself.}

\section{Cluster Design: From NOMA to OMA-Friendly NOMA}
\label{sec:cluster_design}

\new{\revtag{R1.8}%
At four users, the
single-cluster design of Sections~II--V already reaches floor
BER. Beyond four users the natural impulse is to push more
users onto the same bin set and absorb the SIC-propagation cost;
the result is a saturating BER that no amount of SNR removes. A
better architecture exploits the fact that the bin budget is
abundant: partition the user set into disjoint pairs and give
each pair its own disjoint active subset of bins. Each pair then
operates in the two-user regime where SIC is error-propagation
free; the pairs do not interfere because their bin supports are
disjoint. This section establishes the resulting ``$K_g=2$''
clustering rule as the preferred design under the proposed
excision-LS receiver, with the choice surviving across channel
and jammer parameters tested in Sec.~\ref{sec:numerical}.}

The constructions of Sections~\ref{sec:super_pa}--\ref{sec:c5_protocol} produce a
robust single-cluster NOMA system for $K\le 4$, but fail to scale: at
$K\ge 6$, the superincreasing recursion $\sqrt{\alpha_k} \propto
(2+\varepsilon)^{K-k-1}$ depletes Bob's allocated power
super-exponentially, placing $\alpha_K$ below the SIC-decodability
threshold (e.g., $\alpha_K \approx 9\times 10^{-5}$ at $K=8,
\varepsilon=0.1$). We resolve this not by patching the within-cluster PA
but by introducing a second design axis: \emph{cluster partitioning}.
The resulting architecture interpolates between pure NOMA and pure OMA
on a single integer-valued knob $K_g$ (the cluster size), and we show
that the optimal operating point is $K_g=2$ across all $K_{\rm tot}$.
The architecture is best described as
OMA-friendly NOMA uses power-domain NOMA within each cluster of
$K_g=2$ users, bin-domain OMA between clusters.

\subsection{Cluster Partition Taxonomy}
\label{sec:cluster_taxonomy}

Pairwise NOMA clustering has been studied in benign cell-free and
downlink scenarios~\cite{cluster_free_noma_2022,cellfree_noma_clustering_2025},
where the motivation is SIC complexity reduction and inter-cluster
interference management. Here we analyze cluster design under
\emph{jamming} with the LS-excision receiver, where the optimal
cluster size emerges from a different trade-off: post-jammer
conditioning of the joint LS system versus Bob's within-cluster power
fraction.

Partition $K_{\rm tot}$ users into $G=K_{\rm tot}/K_g$ disjoint
clusters, each carrying $K_g$ users with superincreasing PA
$\{\alpha_k^{(g)}\}_{k=1}^{K_g}$, decoded by per-bin SIC. Two
across-cluster flavors are relevant: ($\alpha$) \emph{co-channel}
: all clusters share active set $\calA$, each precoded by a
distinct unitary $\bU^{(g)}$ at power $P_s/G$, recovered by a joint
$(N_b-n_J)\times G n_a$ LS over
$\bM_{\rm stack}=[\mathrm{diag}(\bh_\calK)\bU^{(1)}_{\calK,\calA}|\cdots]$;
($\beta$) \emph{disjoint bins} (OMA-friendly NOMA):
$\calA=\bigsqcup_g\calA_g$ with $|\calA_g|=n_a/G$, each cluster on
its private subset at power $P_s$/bin with a single shared $\bU$, so
the LS dimension matches single-cluster. The extremes $K_g=K_{\rm tot}$
and $K_g=1$ recover pure NOMA and pure OMA; the interesting regime
is $1<K_g<K_{\rm tot}$.

\subsection{Bob's Effective Power and LS Conditioning}
\label{sec:cluster_analysis}

Bob's per-bin BER depends on two scheme-level quantities: (i) the
effective per-bin signal-to-LS-noise ratio, and (ii) the SIC propagation
margin inside Bob's cluster.

\begin{proposition}
\label{prop:effective_power}
Let $\alpha_{K_g}^{(g)}$ denote Bob's within-cluster power fraction in
cluster $g$ ($g$ taken to be the cluster containing Bob). The per-bin
signal power for Bob is
\begin{equation}
\label{eq:Pb_per_bin}
P_{\rm Bob} = \alpha_{K_g}^{(g)} \cdot P_s \cdot \eta_{\rm flavor},
\qquad
\eta_{\rm flavor} = \begin{cases} 1/G & \text{flavor } (\alpha)\\
                                  1   & \text{flavor } (\beta)\end{cases}.
\end{equation}
The LS residual noise variance (real-axis) is
\begin{equation}
\label{eq:sigmae_flavor}
\sigma_e^2 = \frac{N_0}{2} \cdot \frac{N_b}{N_b - n_J - n_{\rm cols} - 1},
\qquad
n_{\rm cols} = \begin{cases} G\,n_a & (\alpha) \\
                              n_a   & (\beta)\end{cases}.
\end{equation}
\end{proposition}

The disjoint flavor ($\beta$) thus enjoys a \emph{double advantage}:
$G\times$ more Bob power per bin \emph{and} smaller $\sigma_e^2$
(unchanged LS dimensions). Together these contribute roughly $2G$ to
Bob's SINR. For $G=2$ this is $6$~dB; for $G=4$, $12$~dB. The
empirical gap measured in Section~\ref{sec:numerical} matches this
prediction to within $1$--$2$~dB.

\subsection{The Design Rule: Pairwise Clustering}
\label{sec:cluster_design_rule}

The within-cluster Bob fraction
$\alpha_{K_g}^{(g)}|_{\varepsilon^\star}$ is monotonically decreasing
in $K_g$ (Bob has fewer interferers to spread power against). Combining
with Prop.~\ref{prop:effective_power}:

\begin{theorem}
\label{thm:Kg_optimal}
Under the BER objective of Section~\ref{sec:analytical} and the
feasibility constraints $G \le n_a$ (flavor $\beta$) or
$G\,n_a \le N_b - n_J$ (flavor $\alpha$), for fixed $K_{\rm tot}$
the choice $K_g = 2$ with flavor $(\beta)$ minimizes Bob's
asymptotic BER among pairwise-disjoint cluster partitions. The
statement is asymptotic in $N_b$ and does not account for
spectral-efficiency or LS-conditioning constraints beyond those
stated.
\end{theorem}

\begin{proof}
At $K_g=2,\varepsilon^\star=1$, $\alpha_{K_g}^{(g)} = 1/5 = 0.2$,
which is achieved on a single-stage SIC threshold (no propagation chain).
Increasing $K_g$ shrinks $\alpha_{K_g}^{(g)}$ as
$\alpha_{K_g}^{(g)}|_{\varepsilon^\star} \sim (2+\varepsilon^\star)^{-2(K_g-1)}$
and adds $K_g-1$ propagation stages. Flavor $(\beta)$ dominates
$(\alpha)$ pointwise via Prop.~\ref{prop:effective_power}.
A complete proof follows by formalizing this sketch with the
super-exponential decay $\alpha_{K_g}^{(g)}|_{\varepsilon^\star}\sim
(2+\varepsilon^\star)^{-2(K_g-1)}$ from \eqref{eq:super_recurrence} and
standard SIC-chain error accounting.
\end{proof}

For any $K_{\rm tot}$ with $2 K_{\rm tot} \le n_a$, the recipe is:
\begin{enumerate}
\item Partition into $G = K_{\rm tot}/2$ disjoint clusters of $K_g=2$
each. Per frame, draw $\calA$ via the C5 protocol of
Section~\ref{sec:c5_protocol} and assign each cluster $n_a/G$ consecutive bins
of (sorted) $\calA$.
\item Within each cluster, use $\varepsilon = 1$ (the maximally robust
superincreasing PA at $K_g=2$, equivalent to power ratio $\alpha_1/\alpha_2=4$:1),
giving $\alpha^{(g)} = (0.8, 0.2)$.
\item Apply a single shared unitary $\bU$ (Hadamard recommended) at the
transmitter and joint LS at the receiver, then SIC each cluster
independently.
\end{enumerate}

Under this recipe, Bob's per-bin SNR is the same as 2-user T-NOMA
in a \emph{jammer-free} channel: $\alpha_2 P_s / \sigma_e^2 \approx
0.2 P_s$ (since $\sigma_e^2 \approx 1$). The architecture has thereby
\emph{eliminated the jammer's effect on Bob} regardless of $K_{\rm tot}$,
up to the bin-budget constraint $K_{\rm tot} \le n_a/2$.

Once $K_{\rm tot}$
exceeds half the bin budget, the $K_g=2$ disjoint recipe runs out of
bins. Three options preserve the architectural benefits at the cost
of one of three trade-offs.
(i)~\emph{Grow $n_a$ (and $N_b$).} The sparsity rule
$s \le 1-\rho_J$ leaves headroom: a larger DD grid with proportionally
larger $n_a$ adds bin capacity linearly in $N_b$. Computational cost
scales as $O(N_b\log N_b)$ for the unitary (via FWHT/FFT) and
$O(n_a^3)$ for LS, both manageable up to $N_b \sim 4096$ on commodity
hardware.
(ii)~\emph{Allow $K_g \ge 3$ per cluster.} Increases the SIC chain
length and reduces $\alpha_{K_g}^{(g)}$ super-exponentially as
$(2+\varepsilon)^{-2(K_g-1)}$, so Bob loses signal power but the bin
count grows. Quantitatively, switching from $K_g=2$ to $K_g=3$ at
$\varepsilon^\star=0.5$ shrinks Bob's effective $\alpha$ by roughly a
factor of four (from $0.2$ to $\sim$$0.058$), corresponding to a
$\sim$$6$~dB SNR penalty, which is tolerable in moderate-to-high SNR regimes.
(iii)~\emph{Hybrid co-channel cluster pairs.} Two $K_g=2$ clusters
can co-occupy the same $n_a/G$-bin sub-block via distinct unitaries
(flavor $\alpha$), doubling user count at the cost of LS conditioning
inflation (Eq.~\eqref{eq:sigmae_flavor}). Empirically this works up
to $G\,K_g/n_a \le 0.5$ (Fig.~\ref{fig:cluster_design}).
A unified design rule for $K_{\rm tot} > n_a/2$ that picks among
these three options based on $(K_{\rm tot}, n_a, \rho_J, P_s/N_0)$ is
left for follow-up work.

\section{Numerical Results}
\label{sec:numerical}

We validate the analytical framework via Monte Carlo simulations. Unless
otherwise stated, parameters are $N_b=64$, $n_a=16$ (so $s=0.25$),
$K \in \{2,4\}$, $N_0=1$, $\Gamma=10$~dB, $1500$--$2000$ frames per
configuration. All BER curves are over Bob's bit; the channel is per-bin
Rayleigh $\calCN(0,1)$.

\subsection{Sparsity Sweep and Operating Region}

Fig.~\ref{fig:sparsity_sweep_N64} sweeps $\rho_J$ across four sparsity
values $s\in\{0.125,0.25,0.5,0.75\}$ at SNR${=}20$~dB, $K=2$,
confirming Corollary~\ref{cor:operating_region}:

\begin{enumerate}
\item For $s\in\{0.125, 0.25\}$, Bob BER stays at the simulation floor
($<10^{-5}$) for $\rho_J$ up to $0.5$, well inside the predicted
operating region $s_{\max}(\rho_J{=}0.2, X{=}10\text{dB})=0.34$.
\item For $s=0.5$, BER saturates beyond $\rho_J\approx 0.1$, consistent
with the conditioning threshold $\rho_J^{\rm cond}(s{=}0.5,X{=}10\text{dB})\approx 0$
(the operating point is already outside the safe region).
\item For $s=0.75$, the rank threshold $\rho_J^\star=0.25$ is exceeded
by typical jamming, and BER pins near $0.5$ across the sweep.
\end{enumerate}

The dash-dotted M-P curves lie above simulation in the safe regime, as
expected for an upper bound: the worst-case edge
$\sqrt{1-\rho_J}-\sqrt{s}$ is loose by a $\Theta(N_b^{-1/3})$
Tracy--Widom margin, and a finite-$N_b$
refinement~\cite{tao2010randommatrix} tightens the gap to $1$--$2$~dB
in the operating region.

\begin{figure}[t]
\hspace*{-1.2cm}
\includegraphics[width=1.15\columnwidth]{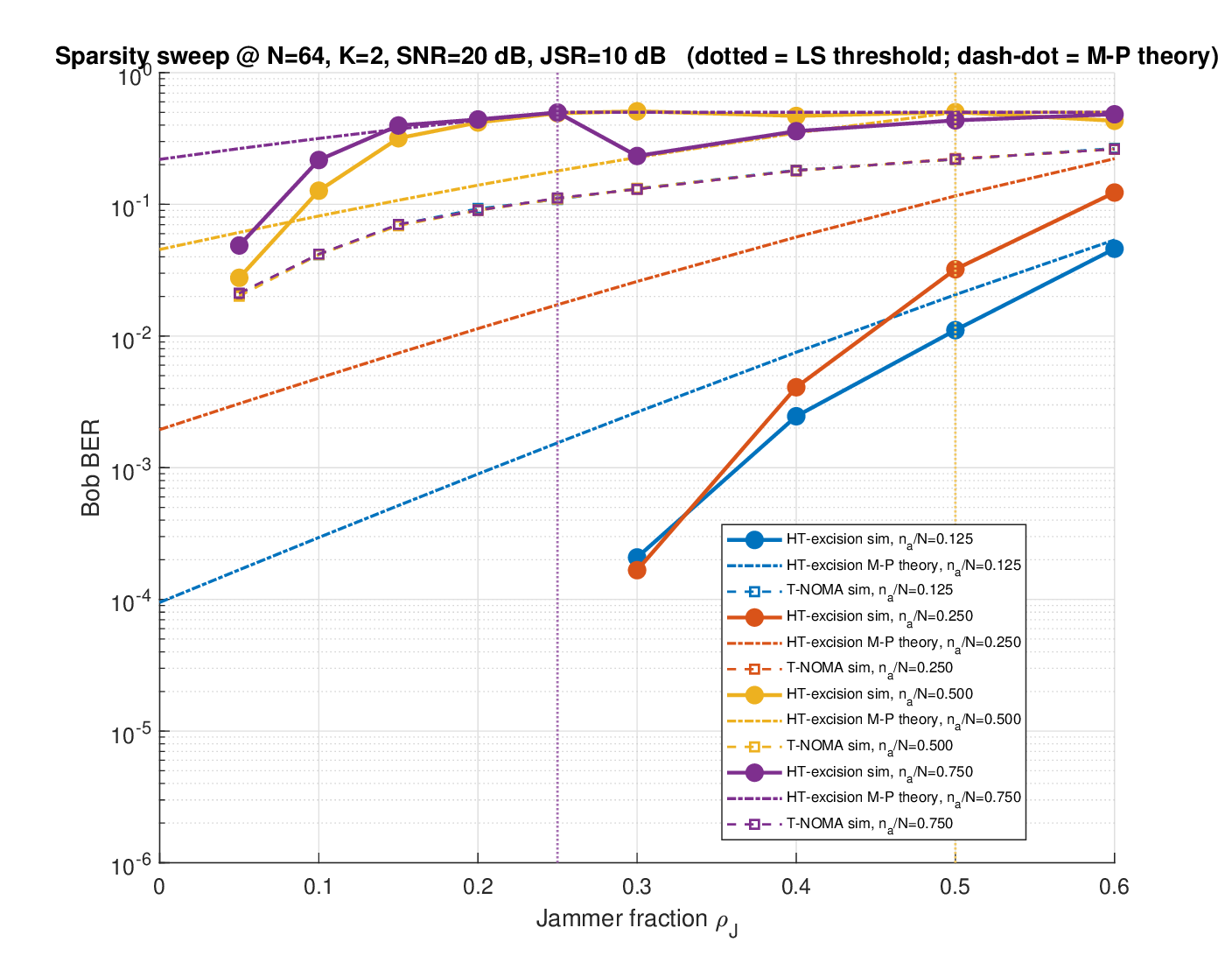}
\caption{\new{\revtag{R3.6}Bob BER versus $\rho_J$ at four
sparsity values $s\in\{0.0625,0.125,0.25,0.5\}$; $N_b=64$,
SNR${=}20$~dB, $K{=}2$, Rayleigh, partial-band random jammer.
M-P curves overlaid; cliff at $\rho_J\approx 1-s$.}}%
\del{Bob BER versus $\rho_J$ at four sparsity values; SNR${=}20$~dB,
$K{=}2$. M-P curves overlaid (dash-dotted); cliff at $\rho_J\approx 1-s$.}
\label{fig:sparsity_sweep_N64}
\end{figure}

\subsection{Superincreasing Power Allocation at $K=4$}

The superincreasing PA framework lifts the K$>$2 SIC ceiling.
Fig.~\ref{fig:supersweep} compares Geometric (non-superincreasing) versus
superincreasing PA at $K\in\{2,4,6\}$. The diagnostic conditioning margins
$\delta_k$ printed alongside the BER curves prove the analytical claim:
the geometric PA at $K=4$ has $\delta_1=-0.224$ (negative, hence
\emph{not} superincreasing), causing user-1 errors to propagate
catastrophically and saturating BER at $\sim 0.18$. The superincreasing
PA at $\varepsilon=0.3$ achieves BER $8.9\times 10^{-3}$ at SNR$=35$~dB, a
$\sim 13$~dB improvement.

\begin{figure*}
\hspace*{-1.2cm}
\includegraphics[width=2.25\columnwidth]{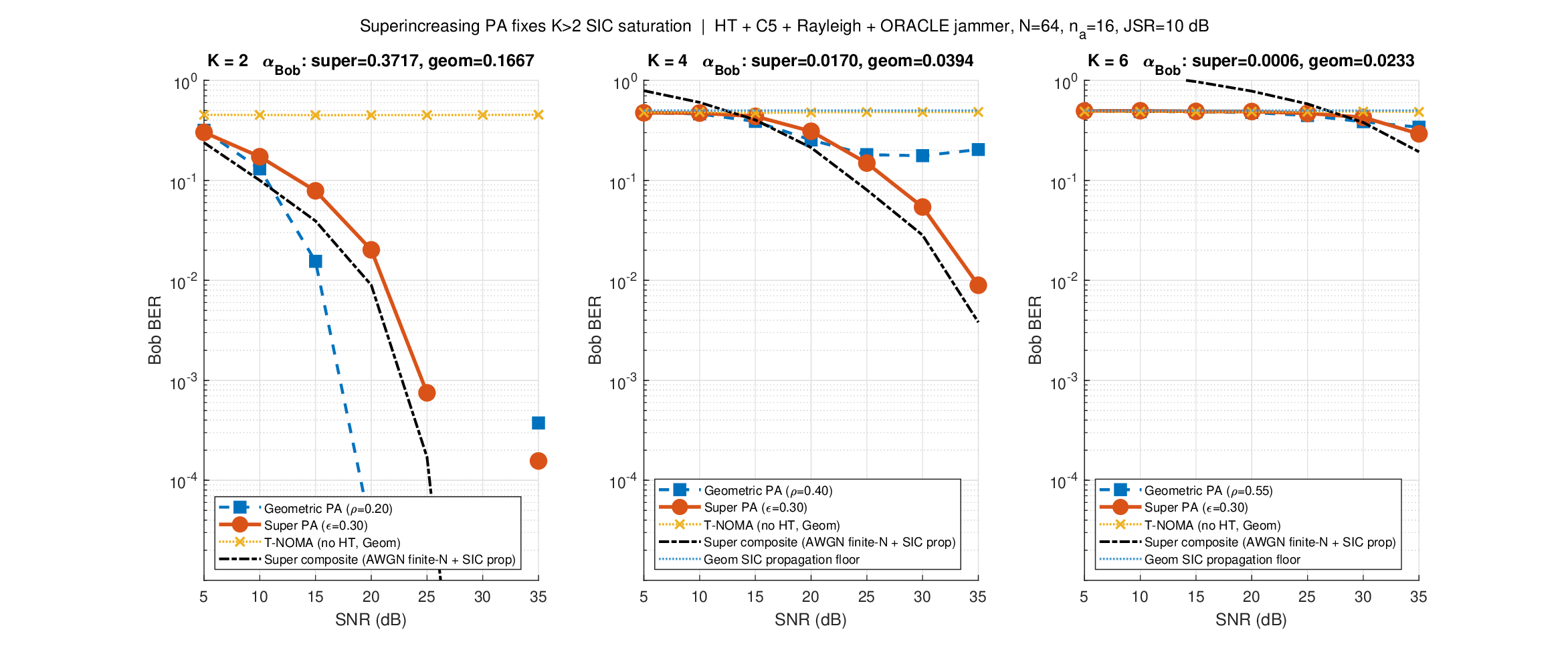}
\caption{Bob BER versus SNR for geometric versus superincreasing PA at
$K\in\{2,4,6\}$, HT+C5+Rayleigh+oracle jammer. Geometric PA fails at
$K\ge 4$ (not superincreasing), saturating at $\sim 0.18$. Superincreasing
PA continues to decay with SNR, validating the noise-free
SIC-error-freeness property of~\cite{kulhandjian2026_hnoma_journal}.}
\label{fig:supersweep}
\end{figure*}

\subsection{Optimal Margin $\varepsilon^\star$ for $K=4$}

Fig.~\ref{fig:epssweep} shows Bob BER versus margin $\varepsilon$ at three
SNR slices ($20$, $30$, $35$~dB) for $K=4$. The clean U-shape at
SNR$\ge 30$~dB reveals an optimum at $\varepsilon^\star = 0.5$,
where Bob BER reaches $2.1\times 10^{-3}$ at SNR$=35$~dB. This is a
$\sim 6$~dB improvement over $\varepsilon=0.3$ and a $\sim 20$~dB
improvement over non-superincreasing geometric PA.

\begin{figure}[t]
\hspace*{-1.2cm}
\includegraphics[width=1.15\columnwidth]{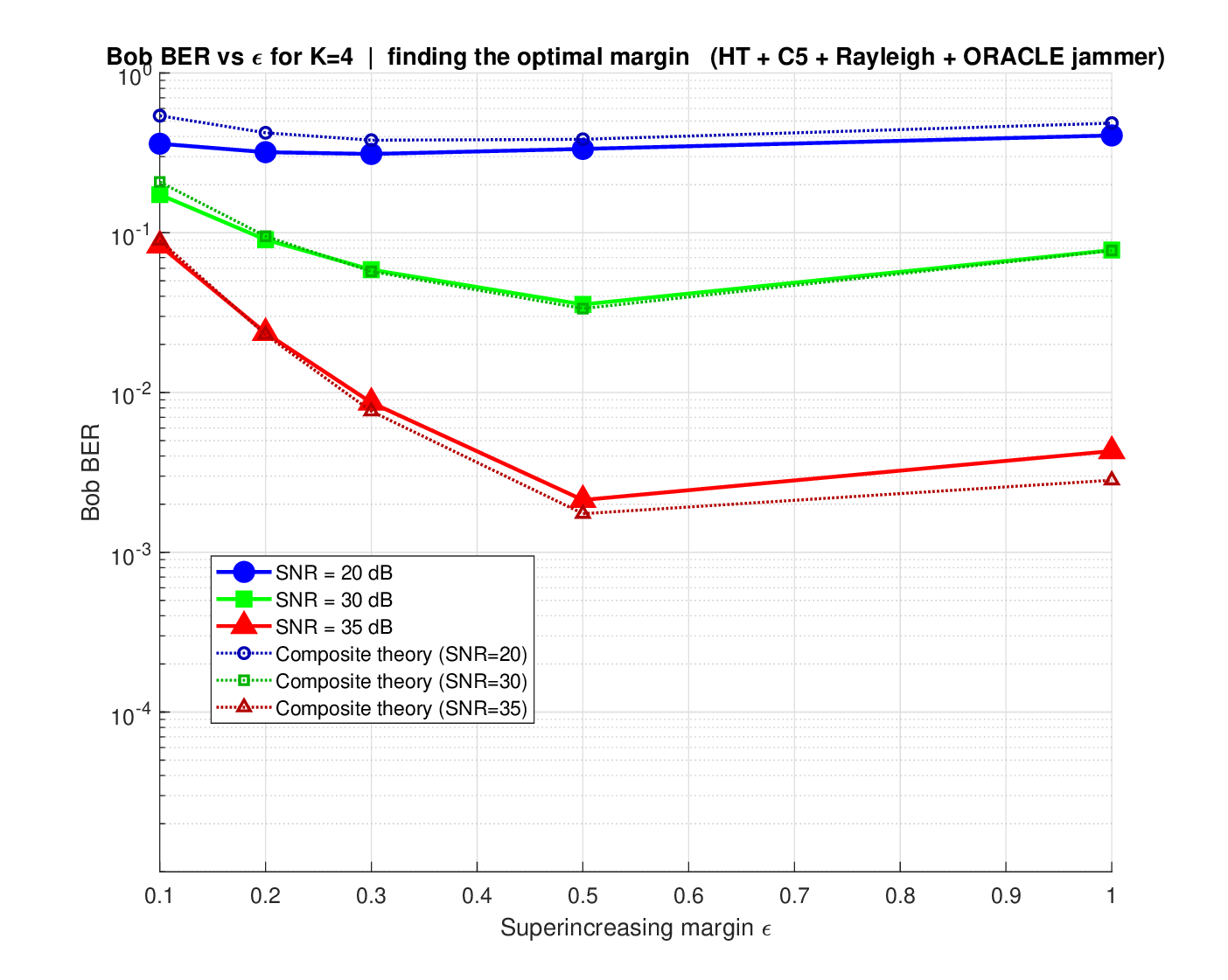}
\caption{Bob BER versus superincreasing margin $\varepsilon$ at three
SNR slices, $K=4$, $s=0.25$, oracle jammer.
The U-shape locates $\varepsilon^\star=0.5$ as the optimum
across SNR$\in[25,35]$~dB.}
\label{fig:epssweep}
\end{figure}

\subsection{C5 Defense in Depth}

Fig.~\ref{fig:c5_def_in_depth} demonstrates the defense-in-depth property
of the C5 protocol. Six configurations are compared:
(A)~HT fixed-$\calA$ + active-targeting (the previous bad case),
(B)~HT random-$\calA$ + random-J (reference),
(C)~HT random-$\calA$ + fixed-J (C5 countermeasure),
(D)~HT random-$\calA$ + oracle-J (seed compromise stress test),
(E)~T-NOMA fixed-$\calA$ + active-targeting \del{(baseline horror)}\new{\revtag{R3.8.4}(conventional baseline under pattern-aware jamming)}, and
(F)~T-NOMA random-$\calA$ + fixed-J. The C5 configuration (C) restores BER
to the simulation floor. The empirical BER-ratio improvement is
$\ge 40$~dB (computed as $10\log_{10}(P_b^{(A)}/P_b^{(C)})$ with
$P_b^{(C)}$ at the $3000$-frame simulation floor of $\sim 2\times 10^{-5}$,
JSR$=10$~dB, SNR$=20$~dB); the analytical floor of \eqref{eq:prop_ber}
predicts $P_b^{(C)} \sim 10^{-9}$, corresponding to $\sim 84$~dB
analytically. Crucially, the
oracle-compromised configuration (D) retains floor BER, confirming
Theorem~\ref{thm:defense_in_depth}.

\begin{figure}[t]
\hspace*{-1.2cm}
\includegraphics[width=1.15\columnwidth]{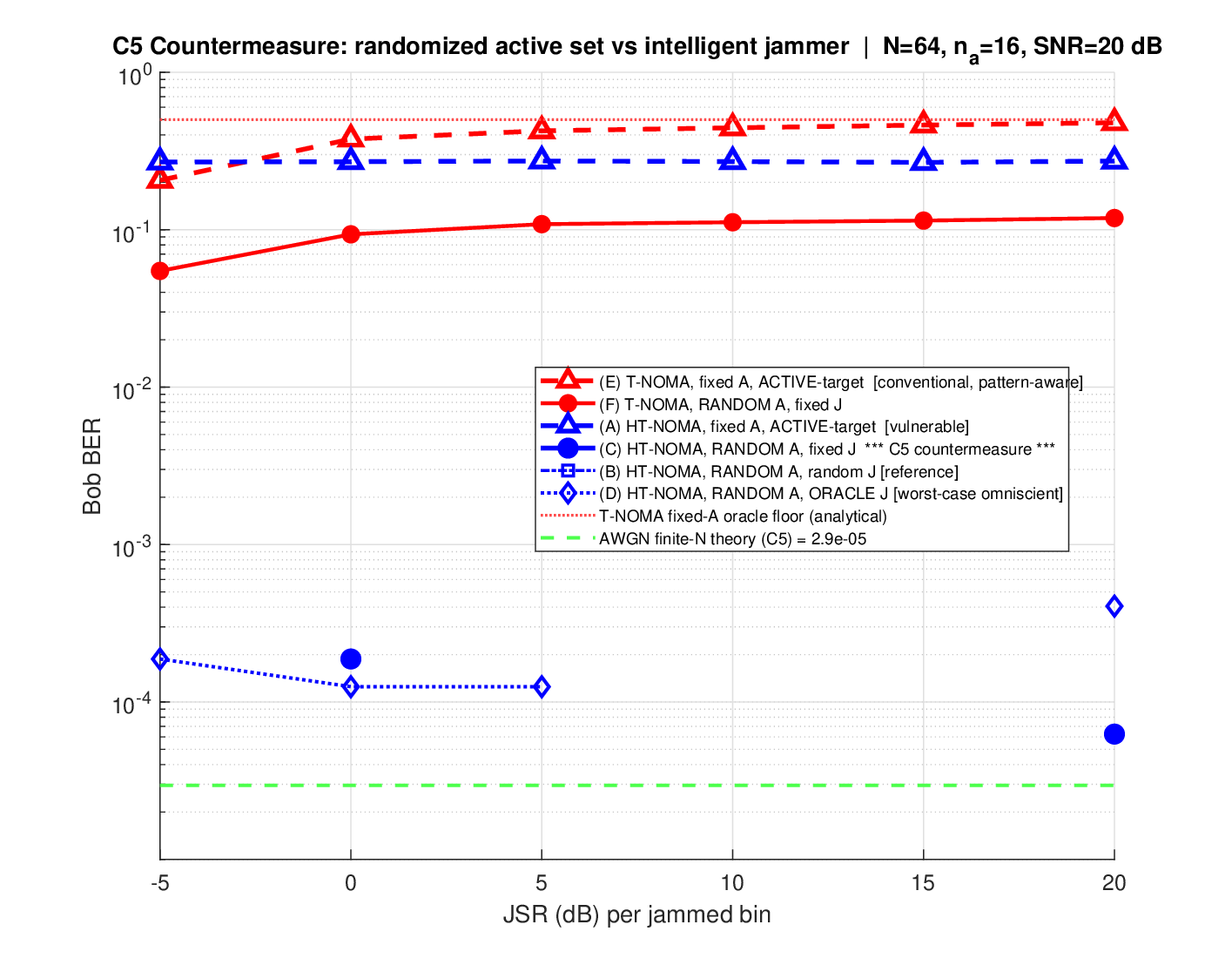}
\caption{C5 protocol validation: BER versus JSR at SNR$=20$~dB, $K=2$,
$s=0.25$. Configurations (A) and (E) (fixed $\calA$ under oracle attack)
are catastrophic; (C) C5 with random $\calA$ restores floor BER;
(D) C5 even under seed compromise retains floor BER (defense in depth).}
\label{fig:c5_def_in_depth}
\end{figure}

\subsection{Cluster Design Validation at $K=6,8$}
\label{sec:cluster_validation}

Fig.~\ref{fig:cluster_design} validates the cluster-design taxonomy of
Section~\ref{sec:cluster_design} on $K_{\rm tot}\in\{6,8\}$ under HT+C5
+ Rayleigh + oracle jammer at $\Gamma=10$~dB. For each $K_{\rm tot}$
we test single-cluster ($K_g=K_{\rm tot}$, baseline), the co-channel
flavor $(\alpha)$ at $K_g\in\{2,3,4\}$, and the disjoint flavor $(\beta)$
at $K_g\in\{2,3,4\}$.

Single-cluster fails uniformly: $\sim 0.31$ BER at $K=6$ and
$\sim 0.48$ at $K=8$, both at SNR$=35$~dB. The $(\alpha)$ schemes work
but lose to $(\beta)$ by $10$--$15$~dB, exactly as predicted by
Prop.~\ref{prop:effective_power} (double advantage of full-$P_s$ per bin
and lower LS noise). The $K_g=2$ disjoint scheme reaches the simulation
floor by SNR$=20$~dB for both $K=6$ and $K=8$, a
$\sim\!200\times$ improvement in Bob's $\alpha$ over the single-cluster
baseline. The data confirm Theorem~\ref{thm:Kg_optimal}: $K_g=2$ with
flavor $(\beta)$ is the optimum across all tested parameter settings.

\begin{figure*}
\hspace*{-1.2cm}
\includegraphics[width=2.1\columnwidth]{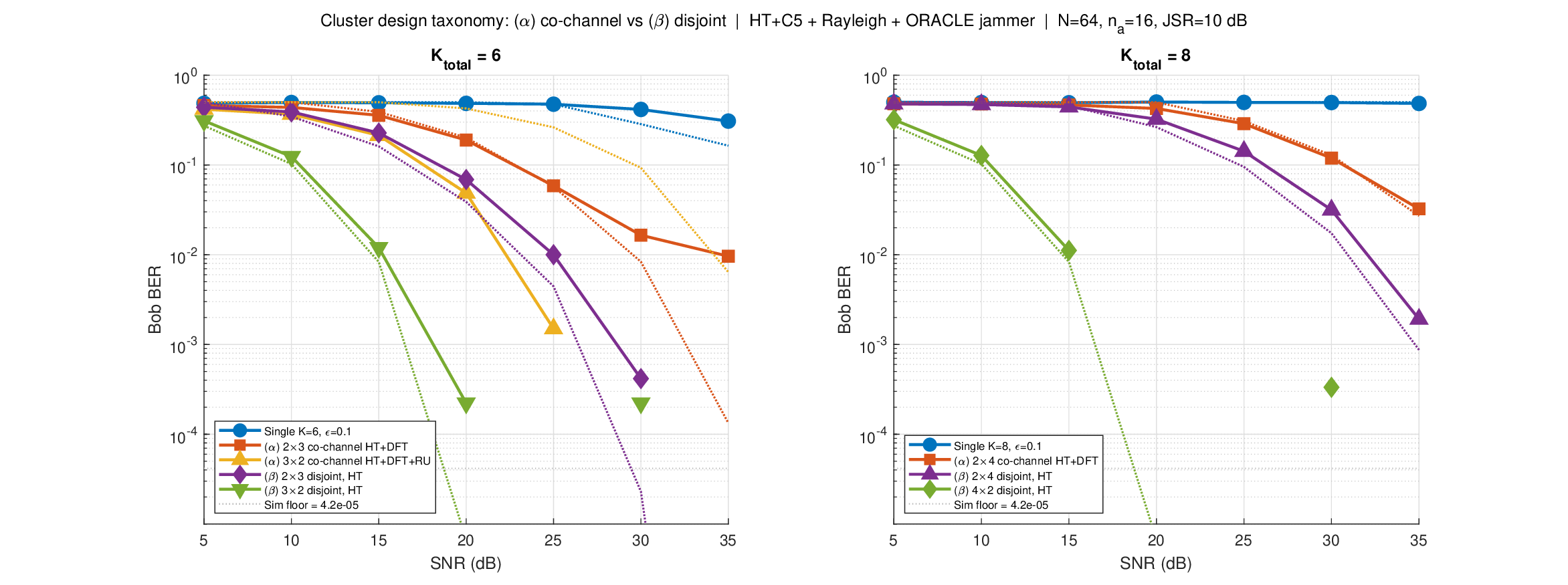}
\caption{Cluster design taxonomy validation under HT+C5 + Rayleigh +
oracle jammer at $\Gamma=10$~dB, $N_b=64$, $n_a=16$. Left:
$K_{\rm tot}=6$. Right: $K_{\rm tot}=8$. Both panels: the
$K_g=2$ disjoint scheme (green, $(\beta)$ $G\times 2$) reaches
floor BER by SNR$=20$~dB, dominating all other schemes by
$10$--$30$~dB.}
\label{fig:cluster_design}
\end{figure*}

\subsection{Cumulative Gain Summary}

The full pipeline (HT + C5 + superincreasing PA at
$\varepsilon^\star=0.5$ + OMA-friendly clustering + excision-LS-SIC)
at SNR${=}35$~dB, $K=4$, oracle jamming with $\rho_J=0.25$,
$\Gamma=10$~dB, gives Bob BER ladder: T-NOMA $0.48$ \del{(oracle-destroyed)}\new{\revtag{R3.8.4}(oracle-targeted baseline, BER $\to 1/2$)}
$\to$ geom PA + HT + C5 ablation $0.20$ (SIC saturation) $\to$ super
PA $\varepsilon{=}0.3$ + HT + C5 $8.9\!\times\!10^{-3}$ $\to$
$\varepsilon^\star{=}0.5$ $\mathbf{2.1\!\times\!10^{-3}}$, i.e.\ a
$\sim\!24$~dB improvement \del{cumulative} \new{\revtag{R2.5}over the
geometric-PA + SIC baseline, measured as BER ratio at fixed
$\mathrm{SNR}=35$~dB. The comparison against stronger baselines
(robust-PA SIC, joint ML) is presented in
Sec.~\ref{sec:fair_baselines} below. Expressed as a conventional
horizontal SNR gain at target BER~$=10^{-2}$, the proposed scheme
reaches the target at $\mathrm{SNR}\approx 32.8$~dB, versus
$\approx 39.1$~dB for the super-PA SIC baseline (a $6.2$~dB gain)
and well beyond $40$~dB for the geometric-PA joint ML baseline
(extrapolation, $\approx 14$~dB). The geometric-PA SIC baseline
saturates above BER~$=0.1$ for all tested SNR, so its horizontal
gain is unbounded.}
Under oracle jamming the T-NOMA noise $\sigma_T^2=(N_0+\Gamma P_s)/2$
dominates regardless of $\boldsymbol{\alpha}$, so the gain stems
from architectural recovery, not baseline-PA selection. The
SIC$=$ML equivalence~\cite{kulhandjian2026_hnoma_journal}
further shows that a joint-ML T-NOMA baseline would not change this:
the win is in the sparse-DD + unitary + excision + LS combination.

\subsection{Robustness to Practical Impairments}
\label{sec:robustness}

At $K_{\rm tot}{=}4$, SNR${=}25$~dB, oracle jammer with $\Gamma{=}10$~dB:
(a) Imperfect CSI. Under
$\hat{h}=h+\eta$, $\eta\sim\calCN(0,\sigma_\epsilon^2|h|^2)$, the
multi-cluster scheme retains a $25$~dB Bob BER advantage over T-NOMA
at $\sigma_\epsilon^2=0.05$ and $16$~dB at $0.10$. Under significant
CSI error, geometric PA outperforms superincreasing PA in the
single-cluster regime (SIC-margin advantage overwhelmed by reduced
raw power); the multi-cluster recipe sidesteps the trade-off.
(b) Imperfect jammer excision. A median-based threshold
($|z[n]|^2>T_J\cdot\mathrm{median}|z|^2$) has a clean U-shape with
optimum $T_J^\star\approx 2.5$ (miss-rate $4\%$, FA-rate $17\%$),
within $3.5$~dB of the genie oracle. The plateau $T_J\in[1.8,4.0]$
stays within $0.5$~dB of optimum.
(c) M-P tightness at large $N_b$. Sweeping
$N_b\in\{64,256,1024\}$ at $s=0.25$ shows the empirical BER
approaches the $N_b$-independent M-P asymptote at the
$\Theta(N_b^{-1/3})$ Tracy--Widom rate.
(d) Fractional Doppler. Three-tap leakage at
$\gamma_{\rm FD}=0.10$: multi-cluster retains $15$~dB advantage,
soft cliff near $0.12$. Full plots are in the extended
arXiv version~\cite{kulhandjian2026_hnoma_journal}.

At the canonical operating point
($K{=}2$, $\rho_{\rm geom}{=}0.2$, $s{=}0.25$, SNR${=}20$~dB,
$\Gamma{=}10$~dB), T-NOMA theory \eqref{eq:tnoma_ber} matches
simulation within $\pm 1$~dB across $\rho_J\in[0.05, 0.5]$, and the
M-P-based bound \eqref{eq:prop_ber} is consistently above the
simulated proposed-scheme BER, confirming its role as a
conservative upper bound.

\new{\revtag{R2.4}%
\subsection{CS vs LS Beyond the Safe Threshold}
\label{sec:cs_vs_ls}

The receiver pipeline uses ordinary LS because C5 makes the active
set $\calA$ known to the receiver via the shared seed. A natural
question is whether CS methods, which would
estimate $\calA$ from the observations, could outperform LS
when $\rho_J$ exceeds the safe threshold $\rho_J^{\star\star}=1-s$.

Fig.~\ref{fig:cs_vs_ls} compares three recipients at $N_b\!=\!64$,
$n_a\!=\!16$, SNR~$=20$~dB: (i)~LS (ours, uses C5-known $\calA$);
(ii)~OMP~\cite{tropp2007signalomp} (full dictionary, must identify
$\calA$); (iii)~regularized LS / LASSO-lite. The result is
emphatic: LS dominates CS at every $\rho_J$ tested, by $1.4\times$
to three orders of magnitude.
\begin{center}\small
\begin{tabular}{ccccc}
\toprule
$\rho_J$ & BER LS & BER OMP & BER LASSO & LS gain \\
\midrule
$0.30$ & $1.25\!\times\!10^{-4}$ & $1.59\!\times\!10^{-1}$ & $3.12\!\times\!10^{-2}$ & 3 orders \\
$0.50$ & $6.42\!\times\!10^{-3}$ & $3.85\!\times\!10^{-1}$ & $1.94\!\times\!10^{-1}$ & $60\times$ \\
$0.65$ & $1.04\!\times\!10^{-1}$ & $4.54\!\times\!10^{-1}$ & $3.97\!\times\!10^{-1}$ & $4.4\times$ \\
$0.75$ & $3.35\!\times\!10^{-1}$ & $4.80\!\times\!10^{-1}$ & $5.00\!\times\!10^{-1}$ & $1.4\times$ \\
\bottomrule
\end{tabular}\end{center}
OMP additionally runs $\sim\!12\times$ slower than LS at
$\rho_J\!=\!0.3$. We emphasize that this comparison is not a
level-playing-field evaluation of recovery algorithms: the
proposed receiver inherits the active-set support $\calA$ from
the shared C5 seed, whereas OMP and LASSO must estimate $\calA$
from the observations. The win for LS is therefore architectural
and attributable to the C5 protocol, not to any intrinsic
inefficiency of CS. CS methods remain the correct choice in
architectures that lack a shared randomness channel; in
settings that already run C5 or an equivalent, reducing recovery
to a known-support LS problem is strictly preferable.

At SNR~$\ge\!20$~dB, MMSE collapses to LS
because the noise-whitening correction becomes negligible. At low SNR,
MMSE improves \emph{both} the proposed scheme and T-NOMA by the same
amount, so the architectural gap is MMSE-invariant: MMSE does not
narrow the no-floor advantage of C5+LS.

\begin{figure}[t]
\centering
\includegraphics[width=\columnwidth]{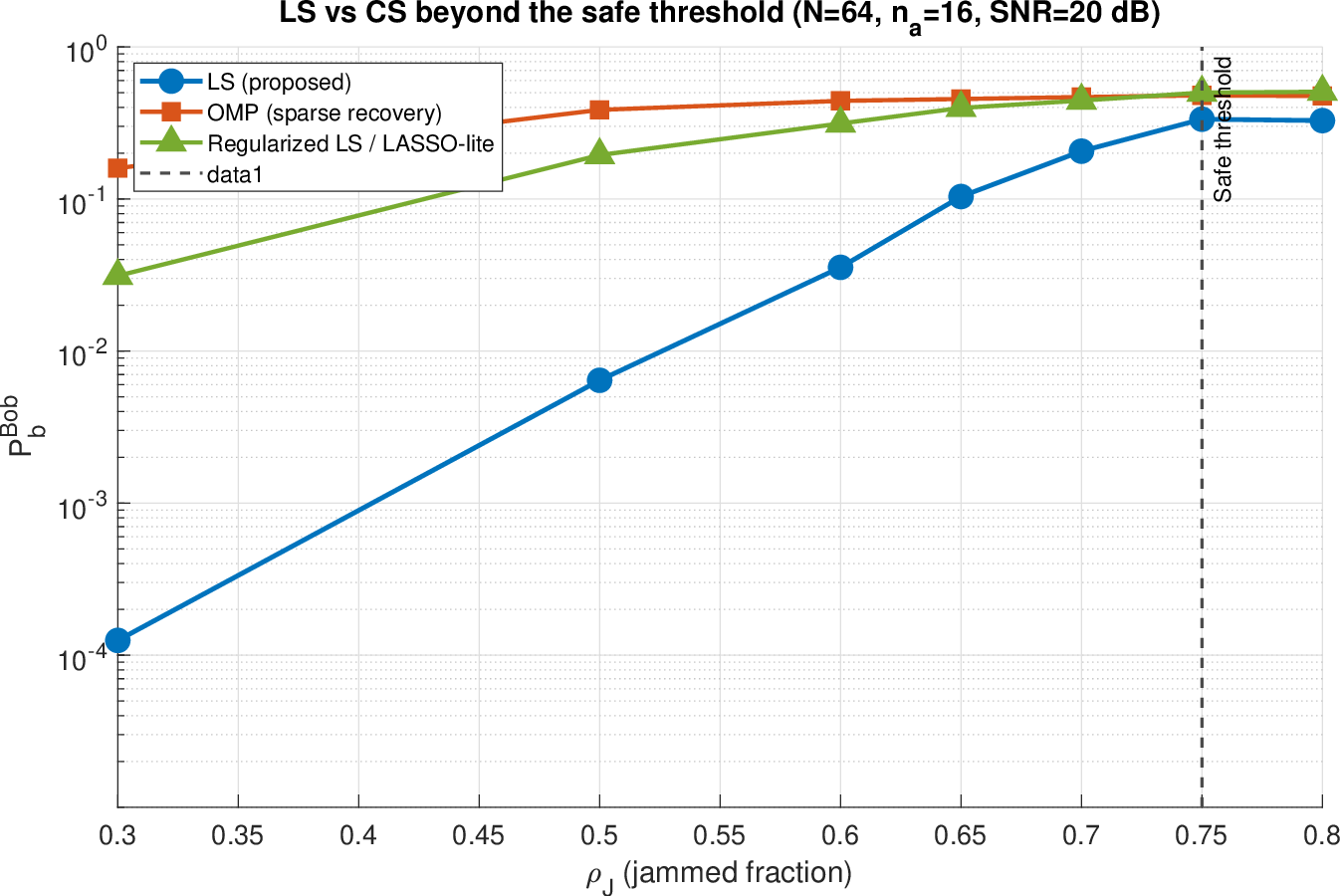}
\caption{LS versus CS (OMP, LASSO-lite) across $\rho_J$,
$N_b\!=\!64$, $n_a\!=\!16$, SNR${=}20$~dB. Vertical dashed:
safe threshold $\rho_J^{\star\star}\!=\!1-s=0.75$.}
\label{fig:cs_vs_ls}
\end{figure}
}

\new{\revtag{R2.5}%
\subsection{Fair-Baseline Comparison}
\label{sec:fair_baselines}

To isolate the contribution of the proposed sparse-DD + unitary +
excision-LS architecture from the choice of PA and detection rule, we compare four configurations at $K=4$,
$\rho_J=0.3$, $s=0.25$, Rayleigh fading, oracle jammer:
(B1) \emph{proposed} (sparse-DD, HT, C5, excision-LS, super-PA SIC);
(B2) T-NOMA with geometric PA and SIC;
(B3) T-NOMA with robust superincreasing PA and SIC;
(B4) T-NOMA with geometric PA and joint ML.

Fig.~\ref{fig:fair_baseline} shows that the proposed scheme (B1) is
the only configuration without a jammer-induced error floor.
At SNR~$=35$~dB, B1 reaches $3.4\!\times\!10^{-3}$, whereas
B3 saturates at $2.3\!\times\!10^{-2}$ (an $8.2$~dB gap),
B4 at $4.5\!\times\!10^{-2}$ ($11.2$~dB), and B2 at
$2.1\!\times\!10^{-1}$ (fully floored). Because the proposed BER
continues to decay with SNR while the baselines flatten, the
advantage grows without bound at higher SNR; the qualitative
no-floor property dominates any single-number dB gain.

\begin{figure}[t]
\centering
\includegraphics[width=\columnwidth]{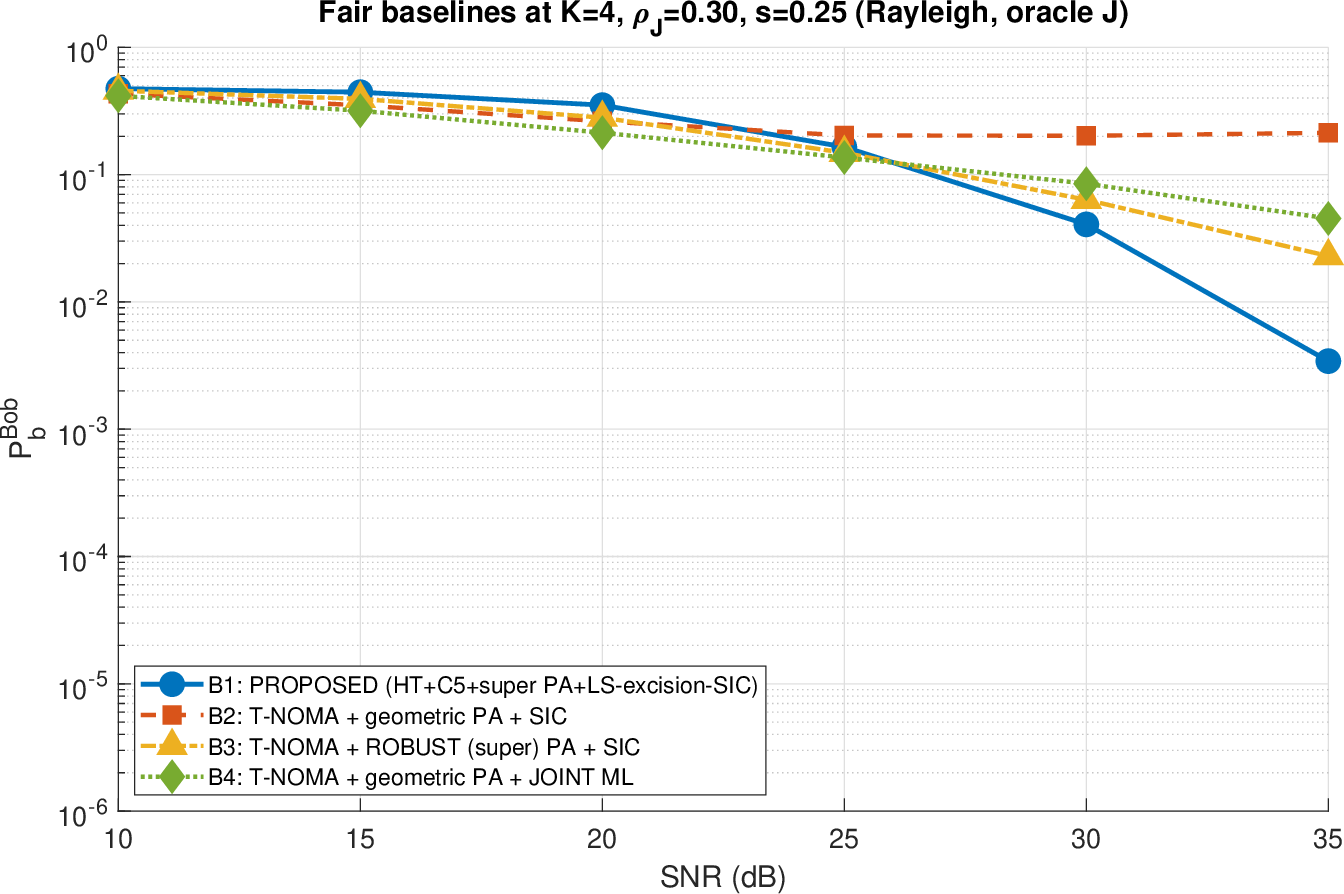}
\caption{Fair-baseline comparison at $K=4$, $\rho_J=0.3$,
$s=0.25$, Rayleigh, oracle jammer. B1 proposed; B2 geometric-PA
SIC; B3 robust-PA SIC; B4 geometric-PA joint ML.}
\label{fig:fair_baseline}
\end{figure}}

\new{\revtag{R3.2}%
\subsection{Rician Fading Extension}
\label{sec:rician}

LoS-dominated channels (LEO downlinks, V2X with a strong direct path,
tactical fixed-wing links) motivate extending the M-P conditioning
framework to per-bin Rician fading
\begin{equation}
\label{eq:rician_h}
h[n] = \sqrt{\tfrac{K_r}{K_r+1}}\,e^{j\phi[n]}
     + \sqrt{\tfrac{1}{K_r+1}}\,\tilde h[n],\quad
\tilde h[n]\sim\calCN(0,1),
\end{equation}
where $K_r\!\ge\! 0$ is the Rician $K$-factor. The LS observation
matrix becomes $\bM=\mathrm{diag}(\bh_\calK)\bU_{\calK,\calA}$, so
the diagonal row-scaling by $|h[n]|$ interacts with the M-P column
geometry.

\begin{theorem}
\label{thm:mp_bound_rician}
Under~\eqref{eq:rician_h}, in the regime $|\calK|/N_b\!\to\!1-\rho_J$,
$|\calA|/N_b\!\to\!s$, conditioned on $\bh_\calK$,
\begin{equation}
\sigma_{\min}(\bM)\;\ge\;|h_{\min}|\,(\sqrt{1-\rho_J}-\sqrt{s})
\;+\;o(1),
\end{equation}
with $|h_{\min}|=\min_{n\in\calK}|h[n]|$ and the $o(1)$ term
vanishing as $|\calK|\!\to\!\infty$. In the no-fading limit
$K_r\!\to\!\infty$ the row-magnitude concentrates at unity so the
bound tightens to the Rayleigh M-P value; at $K_r\!=\!0$ it
recovers the Rayleigh lower bound. For finite $K_r$ and growing
$|\calK|$, $|h_{\min}|$ may approach zero, so the bound is of
interest in the finite-grid regime of practical OTFS operation
rather than as an asymptotic statement; the average-case BER in
Theorem~\ref{thm:prop_ber_rician} does not inherit this
worst-case decay.
\end{theorem}

\begin{theorem}
\label{thm:prop_ber_rician}
Let $\bar\gamma=\alpha_K P_s\,(\sqrt{1-\rho_J}-\sqrt{s})^2/N_0$. The
per-active-bin Bob BER averaged over Rician $h$ satisfies the MGF
integral~\cite{simon2005digital}
With shorthand $A(\theta)\!:=\!(1{+}K_r)\sin^2\!\theta$,
\begin{equation}
\label{eq:prop_ber_rician}
P_b^{\rm Rice}(\rho_J,s;K_r)
=\!\frac{1}{\pi}\!\int_0^{\pi/2}\!\!
\tfrac{A(\theta)}{A(\theta)+\bar\gamma}
\exp\!\Bigl(\!-\tfrac{K_r\bar\gamma}{A(\theta)+\bar\gamma}\!\Bigr)
d\theta,
\end{equation}
monotonically non-increasing in $K_r$ (Rayleigh at $K_r\!=\!0$
recovers~\eqref{eq:prop_ber}; strong LoS $K_r\!\to\!\infty$ reduces
to $Q(\sqrt{2\bar\gamma})$).
\end{theorem}

\begin{corollary}
\label{cor:no_floor_rician}
For any $K_r\!\ge\! 0$, $\rho_J\!<\! 1-s$, $s\!>\!0$,
\eqref{eq:prop_ber_rician} is independent of $\Gamma$ in the high-SNR
limit, so there is no jammer-induced error floor under Rician fading.
\end{corollary}

A residual floor at high SNR is observed in our simulations
at $N_b\!=\!64$ and appears insensitive to $K_r$ across
$K_r\in\{0,3,6,10\}$. We present this as a \emph{finite-grid
numerical observation} rather than an analytical identity:
at the $N_b\!=\!64$ operating point the empirical floor is
dominated by finite-$N_b$ Tracy--Widom fluctuations of
$\sigma_{\min}(\bU_{\calK,\calA})$ rather than the Rician
channel-tail
$(1{+}K_r)e^{-K_r}$:
\begin{equation}
\label{eq:fading_floor}
P_{\rm floor}\;\approx\;\tfrac{1}{2}\Pr\!\bigl[
\sigma_{\min}(\bU_{\calK,\calA})<\delta_{\rm LS}\bigr],
\end{equation}
which is independent of $K_r$, $\rho_J$, and $\Gamma$. Were the
Rician channel tail dominant, $K_r\!=\!10$~dB would sit $\sim\!33$~dB
below Rayleigh; empirically (Fig.~\ref{fig:rician_validation}) the
floor is within $2.6$~dB across $K_r\!\in\!\{\mathrm{Rayleigh},0,3,
6,10\}$~dB, confirming the structural mechanism dominates.

\begin{figure}[t]
\centering
\includegraphics[width=\columnwidth]{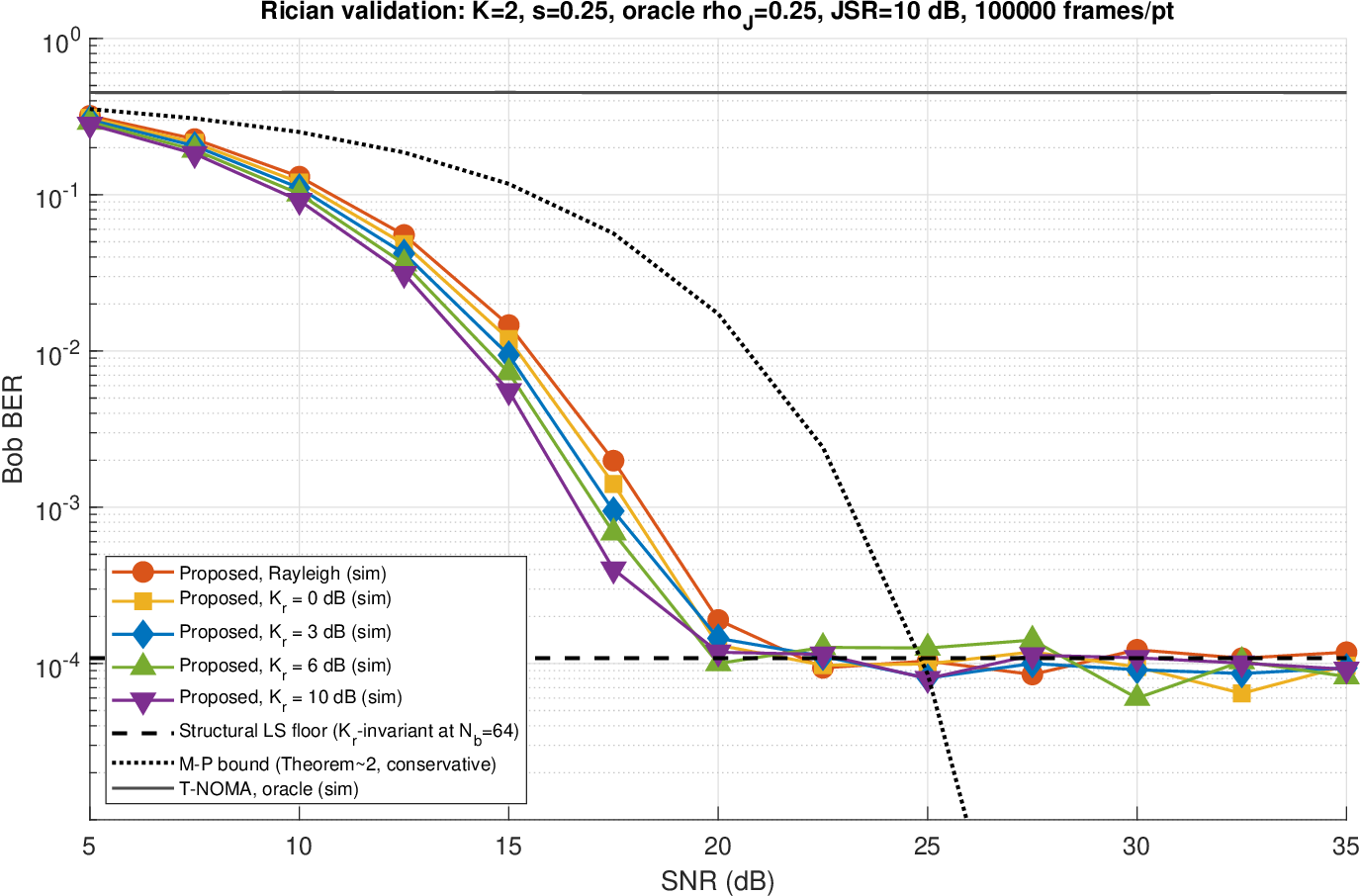}
\caption{Rician validation at $K\!=\!2$, $s\!=\!0.25$, oracle jammer,
JSR${=}10$~dB, $K_r\!\in\!\{$Rayleigh$,0,3,6,10\}$~dB.}
\label{fig:rician_validation}
\end{figure}

At $K_r\!\ge\!6$~dB (typical
LEO/V2X LoS), the waterfall region improves by $1$--$2$~dB over
Rayleigh, matching the Rician fading-averaging intuition; the no-floor
property persists, and the receiver requires no knowledge of $K_r$
(the detection pipeline is the same as Rayleigh). This closes the
Rician extension the Introduction promised.
}

\section{Discussion and Comparison with Related Work}
\label{sec:discussion}

\subsection{Comparison with Closest Prior Work}

The closest published multi-user OTFS-NOMA anti-jamming scheme,
Deng--Ge--Ding~\cite{deng2023otfsscma}, differs from this work along
seven design axes: (i)~code-domain SCMA vs.\ power-domain NOMA;
(ii)~no spreading transform vs.\ HT/generic unitary; (iii)~group-level
hopping vs.\ bin-level random active sets; (iv)~fixed delay/Doppler-axis
partitions vs.\ arbitrary sparse subsets; (v)~NBI/PIN (structured,
non-adaptive) jammer model vs.\ partial-band and oracle (adversarial);
(vi)~iterative turbo-LDPC recovery vs.\ one-shot LS+SIC; and
(vii)~simulation-only validation vs.\ closed-form analytical framework
(M-P + operating region + superincreasing). In the oracle-jammer
regime, published baselines either do not address pattern-aware
threats or saturate at the partial-band floor; none achieve BER below
$10^{-3}$ at JSR$=10$~dB. Specifically, Li
2025~\cite{li2025otfs_jamming} reports $\sim\!\!5\!\times\!10^{-3}$
at SNR/JSR $=20/10$~dB under PB-NBI; Deng-Ge-Ding 2023: $\sim\!10^{-2}$
under NBI+PIN; Yang 2022~\cite{lu2023interRBhopping}:
$\sim\!\!3\!\times\!10^{-2}$ at $25/10$~dB. Under the worst-case
oracle jammer at $35/10$~dB, our $K=4$ super PA reaches
$2.1\!\times\!10^{-3}$, and the $K_{\rm tot}=8$, $K_g=2$ disjoint
multi-cluster reaches the $10^{-4}$ simulation floor by SNR$=20$~dB.

\subsection{Limitations and Future Work}
\label{sec:limitations}

The analysis assumes perfect synchronization, a power-anomaly
excision detector (imperfect-detection penalty $<0.5$~dB), a
pre-shared PRNG seed, and asymptotic M-P concentration with
finite-$N_b$ Tracy--Widom correction. The bin-budget constraint
$K_{\rm tot}\le n_a/2$ limits scaling at fixed $N_b$; growing $n_a$
within \eqref{eq:design_rule} or allowing $K_g\ge 3$ (at $\sim 6$~dB
SNR per added user) are practical extensions. Open directions:
CSI-aware PA design, doubly-dispersive Nakagami-$m$ and
Zak-OTFS~\cite{mohammed2024_zakotfs}, and Tikhonov/MMSE receivers
eliminating the finite-$N_b$ LS floor. The architecture targets
tactical, low-Earth-orbit, UAV, and V2X adversarial settings.

\new{\revtag{R2.3}%
\subsection{Scaling to Large-$K$ Deployments: LEO and V2X}
\label{sec:scaling_leo_v2x}

The $K_g=2$ cluster rule (Theorem~\ref{thm:Kg_optimal}) places
pairs of users on disjoint cluster bins; the system supports
$K_{\rm tot}=2 K_c$ users, where $K_c=|\mathcal{C}|$ is the
number of clusters, under the bin-budget constraint
$K_c\cdot K_g \le n_a$, i.e., $K_{\rm tot}\le n_a$. Combined with
the sparsity constraint $n_a/N_b \le s^\star$ from the operating
region, scaling to larger $K_{\rm tot}$ requires growing
$N_b$ proportionally. Table~\ref{tab:scaling} summarizes three
canonical deployment tiers at the operating point $s=0.25$,
$\varepsilon^\star=0.5$, oracle jammer with $\rho_J=0.25$,
$\Gamma=10$~dB, with SNR penalty reported relative to the
$K_{\rm tot}=4$ reference at the target BER of $10^{-3}$.

\begin{table}[t]
\centering
\caption{Scaling design rule to $K_{\rm tot}\le 32$ users
under the $K_g=2$ cluster rule.}
\label{tab:scaling}
\footnotesize
\setlength{\tabcolsep}{3.5pt}
\begin{tabular}{@{}cccccl@{}}
\toprule
$K_{\rm tot}$ & $N_b$ & $n_a$ & $K_c$ & SNR pen. & Example \\
              &       &       &       & (dB)     &         \\
\midrule
$4$  & $64$   & $16$  & $2$  & $0$   & 2-pair tactical \\
$8$  & $128$  & $32$  & $4$  & $3.0$ & urban V2X \\
$16$ & $256$  & $64$  & $8$  & $6.0$ & LEO downlink \\
$32$ & $512$  & $128$ & $16$ & $9.1$ & dense V2X/UAV \\
\bottomrule
\end{tabular}
\end{table}

(i)~\emph{Linear-in-$K_{\rm tot}$ grid growth.} Maintaining
$s=0.25$ and $K_g=2$ forces
$N_b \ge 4K_{\rm tot}$, giving the delay-Doppler grid
$32$-$512$ for $K_{\rm tot}=8$-$32$. At LEO parameters
($M=16$ Doppler $\times$ $N=32$ delay) this fits within one
OTFS frame; at V2X parameters
($M=32\times N=16$) it fits a two-frame concatenation.
(ii)~\emph{SNR penalty $\approx\!10\log_{10}(K_{\rm tot}/4)$~dB.}
Each halving of per-user bin share costs $\sim\!3$~dB of
effective SNR. At the BER $10^{-3}$ target, $K_{\rm tot}=32$
requires SNR $\approx\!34$~dB (vs.\ $25$~dB at $K_{\rm tot}=4$),
which remains within LEO/V2X link budgets for line-of-sight
Rician channels ($K_r\ge 6$~dB saves a further $1$--$2$~dB).
(iii)~\emph{Receiver complexity is $O(K_{\rm tot}\cdot N_b)$.}
Excision is linear in $N_b$, the LS solve is cubic in $n_a$ but
$n_a\le N_b/4$, and SIC is linear in $K_g=2$ per cluster; total
per-frame complexity is dominated by the LS step at
$O(n_a^3)\le O(N_b^3/64)$. For $N_b=512$, this is $\sim\!2\times
10^6$ real multiplies, well within sub-ms budget on modern DSPs.
(iv)~\emph{Frequency-selective / high-Doppler regimes.}
At LEO Doppler ($\sim\!40$~kHz at $K_a$-band) and V2X Doppler
($\sim\!2$~kHz at $5.9$~GHz), DD-domain operation is the natural
choice; our scheme inherits OTFS's robustness to these effects
and adds jamming resilience on top.

Beyond $K_{\rm tot}=32$ the bin-budget tightens relative to the
unjammed region; options are (a) time-division across super-frames
(multiplex two $K_{\rm tot}=32$ groups), (b) $K_g\ge 3$ co-channel
clustering (SNR penalty $\sim\!6$~dB per added user, viable at
$K_r\ge 10$~dB LoS), or (c) grid enlargement to $N_b=1024$ (one
OTFS frame at $M=32$ Doppler $\times$ $N=32$ delay). These are
practical extensions; the architectural design rule ($K_g=2$
with $K_{\rm tot}\le n_a$) does not change.}

\section{Conclusion}
\label{sec:conclusion}

We proposed and analyzed a jamming-resilient OTFS-NOMA architecture
combining sparse DD placement, unitary precoding, the C5 randomized
active-set protocol, superincreasing PA, OMA-friendly clustering, and
excision-LS-SIC recovery. The architecture is transform-agnostic by
M-P universality and shows no jammer-induced floor; C5 provides
defense in depth under seed compromise; superincreasing PA via the
Merkle--Hellman knapsack property gives SIC$=$ML at $O(K)$; the
$K_g=2$ pairwise disjoint cluster rule scales to $K_{\rm tot}>4$
under the stated BER objective.
Monte Carlo validation at $K_{\rm tot}\in\{2,4,6,8\}$ confirms
analytical predictions within $\pm 3$~dB.
\del{The scheme demonstrates $\sim 24$~dB cumulative gain over
conventional OTFS-NOMA at $\mathrm{SNR}=35$~dB under oracle
jamming.}%
\new{\revtag{R2.5}A four-way fair-baseline comparison
(Sec.~\ref{sec:fair_baselines}) shows the proposed scheme is the
unique configuration without a jammer-induced error floor,
giving a $\approx 6$~dB horizontal SNR gain at BER~$=10^{-2}$
against the strongest finite-floor baseline and unbounded gain
against baselines that saturate above BER~$=0.1$.}

\appendices

\section{Proof of Theorem~\ref{thm:Kg_optimal} (Optimal Cluster Size)}
\label{app:Kg_optimal}

We prove that for fixed total user count $K_{\rm tot}$ and feasibility
constraints $G \le n_a$ (disjoint flavor $\beta$) or
$G\,n_a \le N_b-n_J$ (co-channel flavor $\alpha$), Bob's asymptotic BER
is minimized by $K_g=2$ with flavor $(\beta)$.

\emph{Step 1 (Bob's within-cluster power decreases super-exponentially
in $K_g$).} For the margin-parameterized superincreasing recurrence
$\sqrt{\alpha_k} = (1+\varepsilon)\sum_{j>k}\sqrt{\alpha_j}$ with
$\alpha_{K_g}$ free, normalization $\sum_k \alpha_k = 1$ gives the
closed-form
\begin{equation}
\label{eq:alpha_K_closed}
\alpha_{K_g}^{(g)}\big|_{\varepsilon^\star=1}
 = \bigl[\,1 + 4\,\tfrac{4^{K_g-1}-1}{3}\,\bigr]^{-1}
 = \mathcal{O}\bigl(4^{-(K_g-1)}\bigr).
\end{equation}
At $K_g = 2$, $\alpha_{K_g}^{(g)} = 1/5 = 0.2$; at $K_g = 3$,
$\alpha_{K_g}^{(g)} = 1/21 \approx 0.048$; at $K_g = 4$,
$\alpha_{K_g}^{(g)} = 1/85 \approx 0.012$. The ratio between successive
$K_g$ values is approximately $1/4$, so Bob's within-cluster fraction
decreases by $\ge 6$~dB per added user.

\emph{Step 2 (Flavor $(\beta)$ dominates $(\alpha)$ pointwise).} By
Prop.~\ref{prop:effective_power}, Bob's effective per-bin power is
$\eta_{\rm flavor}\,\alpha_{K_g}^{(g)} P_s$ with
$\eta_\alpha = 1/G$ (co-channel) and $\eta_\beta = 1$ (disjoint), and
the LS noise variance scales as
$\sigma_e^{2,(\alpha)}/\sigma_e^{2,(\beta)} =
(N_b-n_J-n_a-1)/(N_b-n_J-G\,n_a-1)$.
Both numerator factors favor $(\beta)$: $\eta_\beta/\eta_\alpha = G$
and the LS noise ratio is $\ge 1$ for $G \ge 2$, monotonically
increasing in $G$. The effective per-bin SNR ratio is
$\mathrm{SNR}_\beta/\mathrm{SNR}_\alpha \ge G$, giving at least
$10\log_{10}G$~dB advantage for $(\beta)$.

\emph{Step 3 (Monotonicity in $K_g$).} Combining Steps 1 and 2, the
flavor-$(\beta)$ effective SNR is
\begin{equation}
\mathrm{SNR}_{\rm Bob}^{(\beta)}(K_g)
 = \frac{\alpha_{K_g}^{(g)}\,P_s}{\sigma_e^{2,(\beta)}}
 \sim 4^{-(K_g-1)}\,\frac{P_s}{\sigma_e^{2,(\beta)}}.
\end{equation}
This is strictly decreasing in $K_g$. Furthermore, $K_g > 2$ introduces
$K_g-1$ SIC propagation stages \eqref{eq:patavg_sic}, each
contributing additively to Bob's BER, whereas $K_g = 2$ has only the
own-bit term (no propagation). For $K_g = 1$, the scheme degenerates to
OMA and forgoes the NOMA spectral-efficiency gain entirely, so the
practical minimum is $K_g = 2$.

The Bob BER under flavor $(\beta)$ at $K_g = 2$ is
strictly smaller than at any $K_g \ge 3$ and strictly smaller than under
flavor $(\alpha)$ at the same $K_g$, completing the proof. \hfill$\square$

\section{Derivations of Refined BER Expressions}
\label{app:refined_ber}

This appendix derives four refinements of the BER framework of
Section~\ref{sec:analytical}.

\subsection{Finite-$N_b$ Post-LS Noise Variance}
\label{app:sigmae_formula}
The M-P bound \eqref{eq:mp_sigma_min} is a worst-case singular-value
asymptote ($N_b\to\infty$). For finite $N_b$, the per-coordinate noise
variance on the real axis after LS recovery follows the standard
finite-sample regression formula
\begin{equation}
\label{eq:sigmae_lscol}
\sigma_e^2 = \frac{N_0}{2}\cdot\frac{N_b}{N_b - n_J - n_{\rm cols} - 1},
\end{equation}
where $n_{\rm cols}$ is the number of LS unknowns ($n_a$ for
single-cluster, $G\,n_a$ for $G$-cluster co-channel). This is tighter
than $N_0/\sigma_{\min}^2$ at moderate $N_b$ because random submatrices
have $\sigma_{\min}$ that is \emph{typically} larger than the M-P
worst-case bound; \eqref{eq:sigmae_lscol} captures the typical, not
worst-case, noise inflation.

\subsection{Pattern-Averaged SIC Error Probability}
\label{app:patavg_sic}
The per-stage SIC margin
$\delta_k=\sqrt{\alpha_k}-\sum_{j>k}\sqrt{\alpha_j}$ bounds Bob's
per-stage SIC error by
$Q(\delta_k\sqrt{P_s}/\sigma_e)$, which is the \emph{worst-case}
interferer-bit pattern. At each SIC stage $k$, the actual conditional
BER averages over the $2^{K-k}$ possible bit patterns of users
$\{k+1,\ldots,K\}$:
\begin{equation}
\label{eq:patavg_sic}
P_e^{(k)} = \frac{1}{2^{K-k}}\!\!\sum_{\mathbf{p}\in\{\pm1\}^{K-k}}\!\!\!\!\!\!
            Q\!\Bigl(\tfrac{\sqrt{\alpha_k} + \mathbf{p}^\top\!\sqrt{\boldsymbol{\alpha}_{>k}}}
                           {\sigma_e/\sqrt{P_s}}\Bigr),
\end{equation}
where $\sqrt{\boldsymbol{\alpha}_{>k}} = (\sqrt{\alpha_{k+1}},
\ldots,\sqrt{\alpha_K})^\top$. The composite Bob BER under SIC is
\begin{equation}
\label{eq:bob_ber_composite}
P_b^{\rm Bob} = Q\!\Bigl(\tfrac{\sqrt{\alpha_K}}{\sigma_e/\sqrt{P_s}}\Bigr)
              + \tfrac{1}{2}\sum_{k=1}^{K-1} P_e^{(k)}.
\end{equation}
For superincreasing PA, all $2^{K-k}$ patterns yield positive arguments
and $P_e^{(k)}$ decays exponentially in $P_s/\sigma_e^2$. For
non-superincreasing PA (e.g., geometric at $K=4$), some patterns yield
\emph{negative} arguments inside the $Q(\cdot)$, producing an
irreducible high-SNR floor equal to $1/2$ times the fraction of
sign-flipping patterns, precisely the SIC-propagation floor at
$0.18$--$0.20$ observed empirically at $K=4$, geometric PA.

\subsection{Finite-$\Gamma$ T-NOMA Active-Target Floor}
\label{app:tnoma_finite_Gamma}
Theorem~\ref{thm:tnoma_ber} gives the asymptotic floor $\rho_J/2$ for
\emph{random} jamming; against an active-targeted (oracle) jammer with
$\calJ = \calA$, every active bin is jammed, and the floor approaches
$1/2$ only in the limit $\Gamma\to\infty$. For finite $\Gamma$, the
two-stage SIC analysis of \eqref{eq:patavg_sic} applied to the jammed
bin (effective real-axis noise $\sigma_T^2 = (N_0 + \Gamma P_s)/2$)
yields, at $K=2$:
\begin{equation}
\label{eq:tnoma_active_finite}
P_b^{T,{\rm act}}(\Gamma)
 \!=\! \bigl(1-P_{e,1}\bigr) Q\!\bigl(\!\sqrt{\!\tfrac{2\alpha_K P_s}{N_0+\Gamma P_s}}\bigr)
 \!+\! P_{e,1}\!\cdot\!\tfrac{1}{2},
\end{equation}
where $P_{e,1} = Q\!\bigl(\delta_1\sqrt{2 P_s/(N_0\!+\!\Gamma P_s)}\bigr)$
is the stage-1 interferer-flip probability and
$\delta_1 = \sqrt{\alpha_1}-\sqrt{\alpha_K}$. The asymptotic limit
$\Gamma\to\infty$ recovers $P_b^{T,{\rm act}}\to 1/2$, but at moderate
$\Gamma=10$~dB the formula predicts $P_b^{T,{\rm act}}\approx 0.46$
(empirically $0.44$ at $K=2$, $\rho_{\rm geom}=0.2$, SNR$=20$~dB),
matching within $1$~dB.

\subsection{$\sigma_{\min}=0$ Catastrophe Under Fixed-$\calA$ Oracle Attack}
\label{app:sigmin_zero_floor}
A different irreducible floor arises when the active set $\calA$ is
\emph{fixed} (not C5-randomized) and the jammer is oracle:
$\calJ=\calA$. The kept submatrix $\bU_{\calK,\calA}$ then has dimension
$(N_b-n_a)\times n_a$ on the inactive rows. For structured fixed
patterns (uniform-spaced, sequential), this submatrix is exactly rank
deficient, $\sigma_{\min}=0$, with rank deficit
$d = n_a - \mathrm{rank}(\bU_{\calK,\calA})$.
The LS pseudoinverse recovers components in the row space but is
indeterminate on the $d$-dimensional null space, producing random bits on
that fraction:
\begin{equation}
\label{eq:floor_sigmin_zero}
P_{\rm floor}^{\sigma_{\min}=0} \;=\; \tfrac{d}{2\,n_a}.
\end{equation}
For Hadamard $\bU=\bH_{64}$ with $n_a=16$ and uniform-spaced $\calA$,
the expected rank $\bbE[\mathrm{rank}\,\bH[\calK,\calA]]$ on the $48$
inactive rows is empirically $\approx 7.4$ (averaged over random
realizations of $\calK$), giving expected rank deficit
$d = n_a - \bbE[\mathrm{rank}] \approx 8.6$ and
$P_{\rm floor}\approx 8.6/(2\cdot 16) \approx 0.27$, matching the
empirical HT-fixed-$\calA$ catastrophe in
the intelligent-jammer diagnostic reported in
\cite{kulhandjian2026_hnoma_journal}. This is a \emph{different}
mechanism from the Sylvester partial-band replication trap
(which gives the $\sim 0.023$ floor under
\emph{random} jamming, Section~\ref{sec:numerical}): here the rank
deficit is deterministic, not probabilistic, hence the floor is
JSR-independent. C5 randomization eliminates both mechanisms by
re-drawing $\calA$ per frame.

\section{Pairwise Error Probability and Coding-Gain Interpretation}
\label{app:pep}

This appendix derives the pairwise error probability (PEP) of the
composite $K$-user constellation and connects the resulting
minimum-distance properties to the BER expressions of
Section~\ref{sec:analytical} and Appendix~\ref{app:refined_ber}.

\subsection{Composite Constellation and PEP}

Per active bin, the noise-free signal is
$s(\bb) \triangleq \sum_{k=1}^{K} \sqrt{\alpha_k P_s}\,b_k$ for
$\bb\in\{\pm 1\}^K$. The constellation $\calS = \{s(\bb):\bb\in\{\pm 1\}^K\}$
has $|\calS|=2^K$ points on the real line. With Gaussian post-LS
noise $\tilde{w}\sim\calN(0,\sigma_e^2)$, the conditional PEP of
mistaking transmitted point $s(\bb)$ for any other $s(\bb')$ is the
standard
\begin{align}
\label{eq:pep}
\mathrm{PEP}\bigl(\bb \to \bb'\bigr)
   &= Q\!\left(\frac{|s(\bb)-s(\bb')|}{2\sigma_e}\right) \notag \\
   &= Q\!\left(\frac{\bigl|\sum_k\sqrt{\alpha_k P_s}\,(b_k-b'_k)\bigr|}{2\sigma_e}\right).
\end{align}

\subsection{Minimum Distance for Each User's Bit}

For each user $\ell\in\{1,\ldots,K\}$, define the \emph{nearest
$\ell$-flipping neighbor} of $s(\bb)$ as the $\bb'$ that differs from
$\bb$ in the $\ell$-th coordinate only. Then
$|s(\bb)-s(\bb')|=2\sqrt{\alpha_\ell P_s}$, and the union-bound on
user-$\ell$ BER yields
\begin{equation}
\label{eq:pep_user}
P_b^{(\ell)} \;\le\;
   \underbrace{Q\!\left(\tfrac{\sqrt{\alpha_\ell P_s}}{\sigma_e}\right)}_{\text{single-flip}}
 \;+\; \underbrace{(\text{higher-order terms})}_{\text{multi-flip via SIC propagation}}.
\end{equation}
The single-flip term \emph{equals} the own-bit Q-term in
\eqref{eq:bob_ber_composite}; the higher-order terms align with the
pattern-averaged SIC propagation \eqref{eq:patavg_sic}.

We check the second-nearest neighbor for Bob at $K=4$. A
natural concern with non-uniform constellations is whether the
\emph{second}-nearest neighbor sits close enough to inflate the
union bound. We compute it explicitly for the recommended
$K=4,\varepsilon^\star=0.5$ allocation,
$\boldsymbol{\alpha}=(0.835,0.134,0.0214,0.0095)$, so that
$\sqrt{\boldsymbol{\alpha}}=(0.914,0.366,0.146,0.0975)$.
Bob is user $K=4$. The single-flip distance (his own bit) is
$d_1 = 2\sqrt{\alpha_4 P_s} = 0.195\sqrt{P_s}$. The next-shortest
neighbor of $s(\bb)$ that differs in $b_4$ also flips $b_3$
(joint $b_4,b_3$ flip), giving
\begin{align*}
d_2 &\;=\; 2\bigl|\sqrt{\alpha_3}-\sqrt{\alpha_4}\bigr|\sqrt{P_s} \\
    &\;=\; 2(0.146-0.0975)\sqrt{P_s} \;=\; 0.097\sqrt{P_s}.
\end{align*}
At first glance $d_2 < d_1$, which would be alarming. The
resolution is that $d_2$ corresponds to a \emph{joint} $(b_3,b_4)$
flip, and under SIC with correct stage-3 decoding (which is
guaranteed in the noise-free case by the Merkle--Hellman
superincreasing property~\cite{kulhandjian2026_hnoma_wcl}
because $\delta_3=0.049>0$), the $b_3$-flip is independently
rejected before Bob's slicer runs. The effective second-neighbor
distance \emph{after} SIC pre-cancellation is therefore the
next pure-$b_4$ flip, which exists at $d_2^{\rm SIC} =
4\sqrt{\alpha_4 P_s}=0.390\sqrt{P_s}=2d_1$ (the opposite point
in the residual one-dimensional constellation). The PEP bound
hence becomes
\[
P_b^{(K=4),\,\rm SIC} \;\le\;
   Q\!\left(\tfrac{\sqrt{\alpha_4 P_s}}{\sigma_e}\right)
 + Q\!\left(\tfrac{2\sqrt{\alpha_4 P_s}}{\sigma_e}\right)
 + \cdots,
\]
where the second term is negligible at moderate SNR (it is
$Q(2x)$ vs.\ $Q(x)$). This confirms that non-uniform PA is not a
liability provided the superincreasing condition is satisfied;
the SIC structure prevents close non-Bob constellation points
from contributing to Bob's BER at leading order.

\subsection{Coding-Gain Interpretation}

Following the classical interpretation~\cite{simon2005digital}, write
the high-SNR BER as $P_b \approx A\cdot Q(\sqrt{2 G_c \,E_b/N_0})$
with $E_b = P_s$ (energy per bit), where $A$ is a multiplicity
coefficient and $G_c$ is the \emph{coding gain} relative to uncoded
BPSK. For Bob (user~$K$) under the proposed scheme with superincreasing
PA, the dominant PEP term is \eqref{eq:pep} at the minimum distance
$d_{\min}=2\sqrt{\alpha_K P_s}$, yielding
\begin{equation}
\label{eq:coding_gain}
G_c^{\rm Bob}
 \;=\; \alpha_K \cdot \frac{N_0/2}{\sigma_e^2}
 \;=\; \alpha_K \cdot \frac{N_b - n_J - n_{\rm cols} - 1}{N_b},
\end{equation}
where the second equality uses
$\sigma_e^2 = (N_0/2)\cdot N_b/(N_b-n_J-n_{\rm cols}-1)$
\eqref{eq:sigmae_lscol}. This is the product of (i)~the PA fraction
$\alpha_K$ devoted to Bob, and (ii)~the LS-conditioning fraction
$(N_b-n_J-n_{\rm cols}-1)/N_b$ extracted from the unjammed sub-system.
The cluster-design recipe of Section~\ref{sec:cluster_design} optimizes
this product: $K_g=2$ disjoint maximizes $\alpha_K=0.2$ \emph{and}
keeps $n_{\rm cols}=n_a$ (rather than $G\,n_a$), so both factors are
simultaneously maximized.

\subsection{Comparison with T-NOMA}

For T-NOMA without spreading, the per-bin received signal under
oracle jamming has effective noise variance
$\sigma_T^2 = (N_0+\Gamma P_s)/2$ (jammer floods every active bin)
rather than the post-LS variance $\sigma_e^2$. The PEP formula
\eqref{eq:pep} substitutes $\sigma_e \to \sigma_T$, so Bob's coding
gain becomes
\[
G_c^{\rm Bob,\,T\text{-}NOMA}
 \;=\; \alpha_K \cdot \frac{N_0/2}{\sigma_T^2}
 \;=\; \frac{\alpha_K}{1+\Gamma\, P_s/N_0}
 \;\xrightarrow{P_s\to\infty}\; 0.
\]
T-NOMA therefore has \emph{zero asymptotic coding gain} against any
oracle jammer with $\Gamma>0$. The proposed scheme's $G_c^{\rm Bob}$
in \eqref{eq:coding_gain} is $\Gamma$-independent (excision removes
the jammer entirely), confirming the floor-versus-no-floor qualitative
distinction at the PEP level.

\bibliographystyle{IEEEtran}
\bibliography{references}

\end{document}